\documentclass[a4paper,11pt]{amsart}
\usepackage[english]{babel}
\usepackage[centertags]{amsmath}
\usepackage{latexsym, amsfonts, amssymb, graphicx}

\usepackage[all,dvips]{xy}
\usepackage{epic}
\usepackage{xypic}
\usepackage[latin1]{inputenc}
\usepackage{url}
\usepackage{enumerate}

\usepackage{epsfig}
\usepackage{psfrag}
\usepackage{graphicx}

\theoremstyle{plain}
\newtheorem{thm}{Theorem}[section]
\newtheorem{lem}[thm]{Lemma}
\newtheorem{prop}[thm]{Proposition}
\newtheorem{cor}[thm]{Corollary}

\theoremstyle{definition}
\newtheorem{defn}[thm]{Definition}
\newtheorem{nota}[thm]{Notation}
\newtheorem{conj}{Conjecture}

\theoremstyle{remark}
\newtheorem{rem}[thm]{Remark}
\newtheorem{exmp}[thm]{Example}

\def\B{{\mathbf B}}
\def\P{{\mathbf P}}

\def\A{{\mathbf{A}}}
\def\C{{\mathcal{C}}}

\def\PGL{\mathbf{PGL}}
\def\Pc{{\mathcal{P}}}
\def\Bc{{\mathcal{B}}}
\def\wC{{\widetilde{\mathcal{C}}}}
\def\E{{\mathcal{E}}}
\def\I{{\mathcal{I}}}

\newcommand{\F}{\mathbf{F}}

\begin{document}

\title{Incidence structures from the blown--up plane and LDPC codes}

\author{Alain Couvreur}


\thanks{This work was partially supported by the French ANR Defis program under contract
ANR-08-EMER-003 (COCQ project).}


\address{INRIA Saclay \^Ile-de-France, Projet TANC -- CNRS, LIX UMR 7161, \'Ecole Polytechnique, 91128 Palaiseau Cedex, France}
\email{alain.couvreur@inria.fr}


\maketitle

\thispagestyle{empty}


\begin{abstract}
In this article, new regular incidence structures are presented.
They arise from sets of conics in the affine plane blown--up at its rational points.
The LDPC codes given by these incidence matrices are studied.
These sparse incidence matrices turn out to be redundant, which means that their number of rows exceeds their rank.
Such a feature is absent from random LDPC codes and is in general interesting for the efficiency of iterative decoding. 
The performance of some codes under iterative decoding is tested.
Some of them turn out to perform better than regular Gallager codes having similar rate and row weight.
\end{abstract}

\bigbreak



\section*{Introduction}
LDPC codes were first discovered by Gallager in \cite{gal} in the beginning of the sixties. They regained popularity in the mid-nineties and became a fascinating and highly dynamic research area providing numerous applications.
The main reason of this success is that, thanks to iterative decoding, these codes perform very close to the theoretical Shannon limit. For instance, see \cite{RiUrSh}.

\subsection*{Constructions} In the literature, one can distinguish two general approaches to construct LDPC codes.
The first one is based on random constructions (for instance see \cite{McKayNeal}). The second one is based on combinatorial and algebraic methods involving finite fields, block designs, incidence structures and so on.
The present article focuses on the second approach.

A well--known construction of LDPC codes from incidence structures is due to Kou, Lin and Fossorier in \cite{KLF} who proposed to use the points-lines incidence in affine and projective spaces.
Kou et al's approach motivated several other works about the study of LDPC codes arising from incidence structures.
Among them (and the list is far from being exhaustive), other codes from points and lines incidence structures in affine or projective spaces are studied in \cite{Kam07} and \cite{XCDLAG07}.
Codes from partial and semi-partial geometries are considered in \cite{partg} and \cite{semi-part}.
Several other well--known incidence structures have also been used to produce good LDPC codes, among them: generalised quadrangles \cite{KMS}, generalised polygons \cite{LP05}, unital designs \cite{unital}, incidence structures from Hermitian curves \cite{pepe}, oval designs  \cite{Ovals}, flats \cite{TXLAG05} and so on.

\subsection*{Redundant matrices} An interesting property of LDPC codes from incidence structures is that they are frequently defined by parity--check matrices whose number of rows exceeds their rank. Such matrices are said to be \emph{redundant}.
Even if it has no influence on the code, a redundant matrix may improve the efficiency of the iterative decoding.

It is worth noting that 
sparse matrices obtained by random constructions are generically full rank.
Of course, it is always possible to add new rows by linear combinations.
However, a linear combination of a large number of rows is in general non-sparse.
On the other hand, adding linear combinations of a small number of rows does not improve the performance of iterative decoding.

In the case of the codes described in the above-cited references, the parity--check matrix is highly redundant and no row is a linear combination of a small number of other rows. This is particularly interesting for iterative decoding.

\subsection*{New incidence structures} In this article, we  introduce  new incidence structures obtained from the incidence relations between points and strict transforms of conics on the affine plane blown--up at all of its rational points.
Three incidence structures are presented corresponding to three different sets of conics.
These three incidence structures are regular (each point is incident to a constant number of blocks and each block is incident to a constant number of points) and any two points of them have at most one block in common.
Moreover, the girth of their incidence graph is proved to be either $6$ or $8$.

\subsection*{LDPC Codes} Using these incidence structures, we construct {\bf binary} LDPC codes and study their parameters. A formula giving their minimum distance is proved and their dimension is discussed. Two conjectures are stated on the dimensions of some of these codes.
Using the computer algebra software \textsc{Magma}, the actual dimension of some  codes is computed. The information rates of these codes turn out to be close to $1/2$ and their parity--check matrices are highly redundant since they are almost square and hence contain twice more parity checks than necessary.
Finally, simulations of these codes on the Gaussian channel are done and some codes turn out to perform better than regular Gallager codes having the same rates and row weight.

\subsection*{Outline of the article}
The aims of the article are described in Section \ref{aims}.
Some necessary background in algebraic geometry (namely, blow--ups and linear systems of curves) are recalled in Section \ref{ag}. 
Section \ref{SecCon} is devoted to plane conics. Some well-known basic results on conics are recalled and some lemmas used in what follows are proved.
The context and some conventions are stated in Section \ref{context}.
In Section \ref{secconic}, three sets of conics are introduced and studied.
 These three sets are the respective first stones of the constructions of the three new incidence structures presented in the following sections.
In Section \ref{BU}, we introduce the surface $\B$ obtained by blowing up all the rational points of the affine plane.
We derive three interesting sets of curves on this surface arising from the three sets of conics introduced in the previous section.
This yields three new incidence structures
defined in Section \ref{secinc}. They are proved to be regular. Explicit formulas for the number of points per block and the number of blocks incident to a point are given.
The girth of the incidence graph of these structures is computed and proved to be either $6$ or $8$. Moreover, the number of minimal cycles is estimated.
In Section \ref{seccodes}, the LDPC codes from these incidence structures are studied. Some computer aided calculations to get the exact information rate of such codes are presented.
Finally, simulations on the Additive White Gaussian Noise channel are presented at the end of the article.

\subsection*{Notations and terminology} In this article, lots of notations and terminologies are introduced and maintained throughout the paper. To help the reader, an index of notations and terminologies is given in Appendix \ref{AppC}.

\section{Aims of the present article}\label{aims}

The aim of the article is to construct binary LDPC codes which could be efficiently decoded by iterative algorithms.
To seek good candidates, we are looking for binary matrices which
\begin{enumerate}[(1)]
\item\label{CC1} are sparse;
\item\label{CC2} have a Tanner graph with few small cycles and in particular no cycles of length $4$;
\item\label{CC3} are highly redundant, i.e. whose number of rows exceeds the rank.
\end{enumerate}

\medbreak

In order to construct such matrices, we seek incidence structures having some particular properties. First recall the definition of incidence structure.

\begin{defn}[Incidence structure]
  An incidence structure $\I:=(\Pc, \Bc, \mathcal{R})$ consists in three finite nonempty sets. A set of points $\Pc$, a set of blocks $\Bc$ and a set of relations
$\mathcal{R} \subseteq \Pc \times \Bc$ called \textit{incidence relations}.
A point $P\in \Pc$ and a block $B \in \Bc$ are said to be incident if and only if $(P,B) \in \mathcal{R}$. 
\end{defn}

An incidence relation can be described by an incidence matrix $M$, which is a $\sharp \Pc \times \sharp \Bc$ binary matrix such that 
$$
M_{i,j}=\left\{
  \begin{array}{cl}
    1 & \textrm{if}\ (P_i, B_j)\in \mathcal{R}\\
    0 & \textrm{otherwise}
  \end{array}
\right. .
$$

Thus, if we want to use incidence structures to construct codes satisfying the above conditions (\ref{CC1}), (\ref{CC2}), (\ref{CC3}), we have to look for incidence structures 
\begin{enumerate}[(1)]
\item have few incidence relations, i.e. $\sharp \mathcal{R} \ll \sharp \Pc \times \sharp \Bc$;
\item have few small cycles in their incidence graph, in particular no cycles of length $4$ (which means no pairs of points $P,Q$ being both incident with at least two distinct blocks);
\item whose number of blocks exceeds the rank of the incidence matrix with entries in $\F_2$.
\end{enumerate}

\section{Some algebraic--geometric tools}\label{ag}

The aim of this section is to give the minimal background in algebraic geometry to read this article.
Hopefully, the contents of this section are enough to understand what follows.
Most of the proofs are omitted since the presented results are well--known.
We refer the readers to \cite{fulton} and \cite{sch1} for more details.

Most of the results stated in the present section hold for arbitrary fields. However, we chose to state the definitions and the results in the context of the article. Thus, from now on, $\F_q$ denotes some finite field and $\overline{\F}_q$ its algebraic closure.

\subsection{Points, curves, tangent lines and intersections}

\subsubsection{Affine and projective planes}
We denote respectively by $\A^2$ and $\P^2$ the affine and projective plane over $\F_q$. Given a system of homogeneous coordinates $(X,Y,Z)$ on $\P^2$, the projective plane can be obtained as a union of three copies of $\A^2$ corresponding to the subsets $\{X\neq 0\}$, $\{Y\neq 0\}$ and $\{Z\neq 0\}$.
Such subsets are called \textit{affine charts} of $\P^2$.

\subsubsection{Points}
A \emph{geometric point} of $\A^2$ (resp. $\P^2$) is a point whose coordinates are in $\overline{\F}_q$. 
An $\F_q$--\emph{rational point} (or a \emph{rational point}, when no confusion is possible) is a point whose coordinates are in $\F_q$.

\subsubsection{Curves}\label{SecCurves}
\begin{defn}[Curve]
A \emph{plane affine curve} $C$ defined over $\F_q$ (resp. a \emph{plane projective curve} defined over $\F_q$) is the vanishing locus in $\A^2$ (resp. $\P^2$) of a squarefree polynomial $f(x,y)\in \F_q[x,y]$ (resp. a squarefree homogeneous polynomial $f\in \F_q[X,Y,Z]$).
Equivalently, it is the set of geometric points $P\in \A^2$ (resp $\P^2$) with coordinates $(a,b)$ (resp. $(a:b:c)$) such that $f(a,b)=0$ (resp. $f(a,b,c)=0$).
The polynomial $f$ is a \emph{defining polynomial} of the curve. The degree of $C$ is defined as $\deg (C):=deg(f)$. 
\end{defn}

\begin{rem}
  A defining polynomial of a curve over $\F_q$ is unique up to multiplication by a nonzero element of $\F_q$. In what follows we authorise ourselves to say ``the defining polynomial of $C$'' even if it is a misuse of language.
\end{rem}




\begin{defn}[Reducible and irreducible curves]
A curve $C$ is said to be \emph{irreducible} if its defining polynomial is irreducible. Else it is said to be \emph{reducible}.
\end{defn}


\begin{defn}[Smooth and singular points]\label{SmooSing}
Let $C$ be an affine curve, $f$ be its defining polynomial and $P$ be a geometric point of $C$. The curve is said to be \emph{singular} at $P$ if
the partial derivatives $\frac{\partial f}{\partial x}$ and $\frac{\partial f}{\partial y}$ vanish at $P$.
Else, the curve $C$ is said to be \emph{smooth} at $P$.
A curve is said to be \emph{singular} if it is singular at at least one geometric point. Else it is said to be \emph{smooth}.
\end{defn}

\begin{rem}\label{LocToGlob}
  The notion of being smooth or singular at a point $P$ is extended to projective curves by reasoning on an affine chart of $\P^2$ containing $P$.
\end{rem}

\begin{defn}[Projective closure]
  Let $C$ be a plane affine curve defined by the squarefree polynomial $f\in \F_q[x,y]$. Let $f^{\sharp}$ be the unique homogeneous polynomial in $\F_q[X,Y,Z]$ such that $f^{\sharp}(x,y,1)=f(x,y)$ and $\deg(f^{\sharp})=\deg(f)$.
The projective curve of equation $f^{\sharp}=0$ is called the projective closure of $C$. It is obtained by adding to $C$ some \emph{points at infinity} (see Definition \ref{Linfty} further). 
\end{defn}

\subsubsection{Tangent lines}

\begin{defn}[Tangent line]\label{Tangent}
  Let $C$ be a plane affine curve. Let $P$ be a geometric point of $C$ and $L$ be a line of $\A^2$ containing $P$. 
The line $L$ is said to be a \emph{tangent} to $C$ at $P$ if
$$
a\frac{\partial f}{\partial x}(P)+b\frac{\partial f}{\partial y}(P)=0
$$
for all director vectors $(a,b)$ of $L$.
\end{defn}

\begin{rem}
As in Remark \ref{LocToGlob}, the notion of tangent line can be extent to projective curves by considering affine charts.
\end{rem}

\begin{prop}\label{TgtSing}
  A plane curve is smooth at a point $P$ if and only if it has a unique tangent line at this point. 
A plane curve $C$ is singular at $P$ if and only if any line containing $P$ is a tangent to $C$ at $P$.
\end{prop}

\subsubsection{Intersection multiplicity}\label{IntMult}
The notion of intersection multiplicity is pretty easy to feel but heavy to define properly.
Therefore, we do not state its definition for which we refer the reader to \cite{fulton} Chapter 3 \S 3. However, let us give some basic properties of this mathematical function which are enough for what follows.

If $P$ is a geometric point of $\A^2$ (resp. $\P^2$) and $C,D$ are two curves which have no common irreducible component (i.e. their defining polynomials are prime to each other), then one can define a nonnegative integer denoted by $m_P(C,D)$ and called the \emph{intersection multiplicity} of $C$ and $D$ at $P$ which satisfies the following properties.
\begin{enumerate}[(1)]
\item $m_P(C,D)=0$ if and only if one of the curves $C, D$ does not contain $P$;
\item $m_P(C,D)=1$ if and only if both curves contain $P$, are smooth at it and have distinct tangent lines at this point;
\item else, $m_P(C,D)\geq 2$.
\end{enumerate}

\noindent In particular, a curve has always intersection multiplicity $>1$ at $P$ with one of its tangent lines at this point.

To conclude this subsection let us recall the well--known B\'ezout's Theorem.

\begin{thm}[B\'ezout's Theorem]\label{Bezout}
Let $C, D$ be two plane projective curves having no common irreducible components. Then the set of geometric points of intersection of $C$ and $D$ is finite and
$$
\sum_{P\in C\cap D} m_P(C,D)=\deg(C).\deg(D).
$$
\end{thm}

\subsection{Blow--up of a surface at a point}

Blowing up a point of a surface is a classical operation in algebraic geometry.
It is often used to ``desingularise'' a curve embedded in a surface or to ``regularise'' a non regular map at a point.
In this section, we briefly present the notion of blow--up and summarise its most useful properties for the following sections.
We refer the reader to \cite{fulton} Chapter 7 or \cite{sch1} Chapter II.4 for further details.

\begin{defn}[Blow--up of a surface at one point]
Let $S$ be an algebraic surface and $P$ be a smooth point of $S$.
The blow--up of $S$ at $P$ is a surface $\widetilde{S}$ together with a surjective map
$\pi : \widetilde{S} \rightarrow S$ satisfying the following properties.
\begin{enumerate}[(i)]
\item The set $\pi^{-1}(\{P\})$ of pre-images of $P$ by $\pi$ is a curve $E$ isomorphic to the projective line and called the \emph{exceptional divisor}.
\item The restriction $\pi: \widetilde{S}\setminus E \rightarrow S\setminus \{P\}$ is an isomorphism of varieties.
\end{enumerate}
\end{defn}

\begin{prop}
  The blow--up of a surface at a point is unique up to isomorphism.
\end{prop}

\begin{exmp}[Blow--up of the affine plane]
  Consider the affine plane $\A^2$ with coordinates $(x,y)$ and let $P$ be the origin.
Then, the blow--up of $\A^2$ at $P$ is the surface 
$$
\widetilde{\A}^2:=\{(x,y,(u:v))\in \A^2\times \P^1\ |\ xu=yv\}.
$$
together with the projection map
$$
\pi:\left\{
  \begin{array}{ccc}
    \widetilde{\A}^2 & \rightarrow & \A^2 \\
    (x,y,(u:v)) & \mapsto & (x,y)
  \end{array}
\right. .
$$
One sees easily that the set of pre-images of the origin is isomorphic to $\P^1$ and
that any point of $\A^2 \setminus \{P\}$ has a unique pre-image by $\pi$.
\end{exmp}


\begin{defn}[Strict transform of a curve]
Let $S$ be a smooth surface and $P$ be a point of $S$.
  Let $\pi: \widetilde{S} \rightarrow S$ be the blow--up of $S$ at $P$ and denote by $E$ the corresponding exceptional divisor.
Let $C$ be a curve embedded in $S$ and containing $P$.
The decomposition into irreducible components of the algebraic set $\pi^{-1}(C)$ is of the form $\pi^{-1}(C)=E\cup \widetilde{C}$, where $\widetilde{C}$ does not contain $E$.
The curve $\widetilde{C}$ is called the \emph{strict transform} of $C$ by $\pi$.

If $C$ does not contain $P$, its strict transform is defined as $\widetilde{C}:=\pi^{-1}(C)$.
\end{defn}

\begin{rem}
  The above definition extends naturally to a map obtained by the composition of a finite number of blow--ups.
\end{rem}

In the proposition below, we summarise most of the properties of blow--ups needed in what follows.

\begin{prop}\label{summarize}
  Let $\pi: \widetilde{S} \rightarrow S$ be the blow--up of a surface $S$ at a smooth point $P$. Denote by $E$ the exceptional divisor.
  \begin{enumerate}[(i)]
  \item\label{sum1} There is a one-to-one correspondence between tangent lines to $S$ at $P$ and points of the exceptional divisor.
In particular, given a tangent line $L$ to $S$ at $P$, there exists a unique point $Q\in E$ such that for all curve $C\subset S$ smooth at $P$ and tangent to $L$ at this point, the strict transform $\widetilde{C}$ meets $E$ at $Q$ and only at this point.
  \item\label{nointer} If two curves $C,D$ meet at $P$, are smooth at it but have no common tangent line at $P$, then their strict transforms do not meet in a neighbourhood of the exceptional divisor.
  \item\label{IncMult} If two curves $C,D$ meet at $P$, are smooth at it and have a common tangent $L$ (i.e. their intersection multiplicity at $P$ is greater than or equal to $2$), then, their strict transforms meet at the point $Q\in E$ corresponding to $L$. Moreover,
$$
m_P(C,D)> m_Q(\widetilde{C}, \widetilde{D})\geq 1,
$$
where $m_P( .\ \!,.)$ denotes the intersection multiplicity at $P$.
  \end{enumerate}
\end{prop}

We conclude this sub-section with the following lemma.

\begin{lem}\label{same}
  In the context of Proposition \ref{summarize}, if $C\subset S$ is a curve which is smooth at $P$ or avoids $P$, then $\widetilde{C}$ is isomorphic to $C$.
In particular, $C$ and $\widetilde{C}$ have the same number of rational points.
\end{lem}

\subsection{Linear automorphisms of the projective plane}
Some proofs in this article involve the action of the group $\PGL (3, \F_q)$ of linear automorphisms of $\P^2$. It is well--known that this group acts simply transitively on $4$--tuples of rational points of $\P^2$ such that no $3$ of them are collinear. More generally we have the following lemma.

\begin{lem}\label{Aut}
The group $\PGL (3, \F_q)$ acts transitively on $4$--tuples of the form:
\begin{enumerate}[(1)]
\item\label{Aut1} $(P, \overline{P}, P_3, P_4)$ where $P_3, P_4$ are rational points and $P, \overline{P}$ are non rational points conjugated under the action of the Frobenius and no $3$ of these points are collinear;
\item \label{Aut2} $(P_1, P_2, L_1, L_2)$, where $P_1,P_2$ are rational points and $L_1,L_2$ are lines defined over $\F_q$ such that $P_1\in L_1$, $P_2\in L_2$, $P_1\notin L_2$ and $P_2\notin L_1$;
\item\label{Aut3} $(P_1, P_2, P_3, L)$ such that $P_1, P_2, P_3$ are non-collinear rational points and $L$ is a line defined over $\F_q$ containing $P_3$ and avoiding $P_1, P_2$;
\item\label{Aut4} $(P, \overline{P}, P_3, L)$ where $P_3$ is a rational point and $P, \overline{P}$ are conjugated under the action of the Frobenius map, the three points are non collinear and $L$ is a line defined over $\F_q$ containing $P_3$ and avoiding $P, \overline{P}$.
\end{enumerate}
In case (\ref{Aut1}), the action is also free.
\end{lem}

The proof of Lemma \ref{Aut} is given in Appendix \ref{AppA}.

\subsection{Linear Systems of plane projective curves}
Linear systems of curves is a central object in this article.
Let us recall their definition and some of their properties. 
For reference, see \cite{fulton} Chapter 5 \S 2.

\begin{defn}[Linear System of curves]\label{Slin}
  A linear system $\Gamma$ of curves in the affine (resp. projective) plane is a set of (possibly non-reduced) curves linearly parametrised by some projective space $\P^n$.
That is, there exists a family of linearly independent polynomials $F_0, \ldots , F_s \in \F_q[x,y]$ (resp. linearly independent homogeneous polynomials in $\F_q[x,y,z]$ of the same degree), such that for each element $C\in \Gamma$, there exists a unique point $P=(p_0: \ldots : p_s)\in \P^s$ such that $p_0F_0+\cdots +p_sF_s=0$ is an equation of $C$.
\end{defn}

\begin{rem}[Non-reduced curves]\label{nonRed}
  In \S \ref{SecCurves}, a curve is defined as the vanishing locus of a squarefree polynomial. In the above definition, some elements of the linear set of polynomials may have square factors.
The good formalism to take care of this difficulty is that of Grothendieck's schemes (see for instance \cite{H} Chapter II). However, this theory requires a huge background which is useless for what follows.

In the present article, the linear systems are always linear systems of curves of degree $2$. The only degenerate cases are polynomial of the form $l(x,y)^2$, where $l$ has degree $1$. In this case, the corresponding ``curve'' $C$ is called the \emph{double line} supported by $L$ (where $L$ is the line of equation $l(x,y)=0$).
The points of $C$ are those of $L$. In addition, $C$ is singular at all of its points and any line containing a point $P\in C$ is tangent to $C$ at $P$. Considering $l(x,y)^2$ as the defining polynomial of $C$, this property is actually coherent with Definition \ref{Tangent} and Proposition \ref{TgtSing}.
\end{rem}

\begin{defn}
  In the context of Definition \ref{Slin}, if $P=(p_0: \ldots : p_s)$ is a geometric point of $\P^s$, the curve of equation $p_0F_0+\cdots +p_sF_s=0$ is called a \emph{geometric element} of $\Gamma$. If $P$ is a rational point of $\P^s$, then this curve is said to be a \emph{rational element} of $\Gamma$. The rational elements of $\Gamma$ are the curves of $\Gamma$ which are defined over $\F_q$.
\end{defn}

\begin{defn}
  The dimension of a linear system is the dimension of its projective space of parameters.
\end{defn}

The example of the linear system of conics is studied in the following section.

\section{The linear system of plane conics}\label{SecCon}
This section is devoted to the linear system of plane projective conics and the properties of some of its subsystems.
The family of polynomials $X^2, Y^2, Z^2, XY, XZ,$ $YZ$ generates a linear system of dimension $5$ on $\P^2$ called the \emph{linear system of plane conics}. Its elements are classified in the following proposition.

\begin{prop}[Classification of plane projective conics]\label{Classification}
 An element of the linear system of conics can be
  \begin{enumerate}[(1)]
  \item\label{irred} either a smooth irreducible curve, in this situation, it has $q+1$ rational points;
  \item\label{Red2} or a union of two lines defined over $\F_q$;
  \item\label{Red3} or a union of two lines defined over $\F_{q^2}$ and conjugated under the Frobenius map;
  \item or a ``doubled'' line defined over $\F_q$ (see Remark \ref{nonRed}).
  \end{enumerate}
\end{prop}

The two following lemmas are frequently useful in what follows. To state them, the following notation is convenient.

\begin{nota}\label{notdtes}
  Let $P,Q$ be two points of the affine (resp. projective) plane.
  We denote by $(PQ)$ the unique affine (resp. projective) line joining $P$ to $Q$.
\end{nota}

\begin{lem}\label{contrick}
  Let $P_1, P_2, P_3$ be three non-collinear geometric points of $\P^2$.
Let $L$ be a line containing $P_3$ and avoiding $P_1$ and $P_2$.
The linear system $\Lambda_1(P_1, P_2, P_3, L)$ of conics containing $P_1, P_2, P_3$ and tangent to $L$ at $P_3$ has dimension $1$.  
Moreover its only singular geometric elements  are $C:=(P_1 P_2)\cup L$ and $C':=(P_1 P_3) \cup (P_2 P_3)$.
\end{lem}

\begin{rem}
  In the above statement, the curve $C':=(P_1 P_3) \cup (P_2 P_3)$ is singular at $P_3$. Therefore, from Proposition \ref{TgtSing}, any line containing $P_3$ is tangent to $C'$ at $P_3$. 
\end{rem}

\begin{proof}[Proof of Lemma \ref{contrick}]
Applying a suitable automorphism in $\PGL(3, \overline{\F}_q)$ (use Lemma \ref{Aut} (\ref{Aut3}) replacing $\PGL (3,\F_q)$ by $\PGL (3,\overline{\F}_q)$), one can assume that $P_1, P_2, P_3$ have respective coordinates $(1:0:0)$, $(0:1:0)$, $(0:0:1)$ and the line $L$ has equation $Y=X$.
Take an equation $aX^2+bXY+cY^2+dXZ+eYZ+fZ^2=0$ of a conic.
The vanishing conditions at $P_1, P_2, P_3$ yield respectively $a=0, c=0$ and $f=0$. The tangency condition entails $d=e$.
Thus, the resulting linear system is parametrised by $\P^1$ and generated by the polynomials $XY$ and $XZ+YZ$.

Now, let $D$ be a singular element of $\Lambda_1 (P_1, P_2, P_3, L)$. From Proposition \ref{Classification}, $D$ is a union of two lines. We conclude by noticing that the only pairs of lines satisfying the conditions of the linear system are $(P_1 P_2) \cup L$ and $(P_1 P_3) \cup (P_2 P_3)$.
\end{proof}

\begin{lem}\label{contrick2}
  Let $P_1, P_2$ be two geometric points of $\P^2$. Let $L_1, L_2$ be two lines such that $L_1\ni P_1$, $L_2 \ni P_2$, $P_1 \notin L_2$ and $P_2 \notin L_1$.
Then the linear system $\Lambda_2(P_1, P_2, L_1, L_2)$ of conics containing $P_1, P_2$ and being respectively tangent to $L_1, L_2$ at these points has dimension $1$.
Moreover its only singular geometric elements are $L_1 \cup L_2$ and the \textit{doubled line} supported by $(P_1 P_2)$ (see Remark \ref{nonRed}).
\end{lem}

\begin{proof}
  It is almost the same approach as that of the proof of Lemma \ref{contrick}
\end{proof}

\begin{rem}\label{cardi}
  In Lemmas \ref{contrick} and \ref{contrick2}, the linear systems have exactly $q-1$ smooth $\F_q$--rational elements.
\end{rem}

\section{Context, notations and terminology}\label{context}
In what follows, the cardinal $q$ of the base field $\F_q$ is assumed to be greater than or equal to $4$. The characteristic of the base field may be odd.
We fix a system of coordinates $(x,y)$ for $\A^2$ and a system of homogeneous coordinates $(X:Y:Z)$ on $\P^2$. Moreover, we identify $\A^2$ as an affine chart of $\P^2$ by the map $(x,y)\mapsto (x:y:1)$.

\medbreak

\noindent \textbf{Caution.} In this whole article, we deal with error correcting codes and with algebraic geometry over finite fields. It is worth noting that, although the geometric objects we deal with are defined over finite fields $\F_q$ with $q\geq 4$ and possibly odd, all the codes we construct are \textbf{binary} codes (i.e. defined over $\F_2$).

\begin{nota}\label{Linfty}
The line $\{Z=0\}$, is called ``the line at infinity'' and denoted by $L_{\infty}$. 
We fix an element $\alpha \in \F_{q^2} \setminus \F_q$ and denote respectively by $P_{\infty}$, $Q_{\infty}$, $R_{\infty}$ and $\overline{R}_{\infty}$ the points of respective coordinates:
$$
\begin{array}{ccccccc}
  P_{\infty} & := & (0:1:0) & \qquad & Q_{\infty} & := & (1:0:0) \\
  R_{\infty} & := & (\alpha : 1 : 0) & \qquad & \overline{R}_{\infty} & := &
(\alpha^q : 1 : 0). 
\end{array}
$$
The points $R_{\infty}$ and $\overline{R}_{\infty}$ are non rational but conjugated under the Frobenius action.
\end{nota}

\begin{defn}[Vertical and horizontal lines]\label{vertical}
We call vertical (resp. horizontal) lines the affine lines having an equation of the form $x=a$ (resp. $y=a$), where $a\in \F_q$.
Equivalently, vertical (resp. horizontal) lines are affine lines whose projective closure contain the point $P_{\infty}$ (resp. $Q_{\infty}$).
\end{defn}

We keep Notation \ref{notdtes}: given two points $P,Q$ we denote by $(PQ)$ the line joining these points. Moreover, we introduce the following notation.

\begin{nota}\label{nottgtes}
  Let $C$ be a plane curve and $P$ be a smooth point of it, we denote by $T_PC$ the tangent line of $C$ at $P$.
\end{nota}

\section{Incidence structures of conics of the affine plane}\label{secconic}

In the present section, we introduce the sets of conics $\C_1 (q), \C_2 (q)$ and  $\C_3(q)$ which are further used to construct the block set of the incidence structures introduced in \S \ref{subsecinc}.

\subsection{The problem of small cycles}
As said in \S \ref{aims}, one of our objectives is to construct incidence structures in which two points are both incident with at most one block.
If one considers a set of points of the affine or projective plane together with a set of conics, the corresponding incidence structure does in general not satisfy such expectations. 
Indeed, from B\'ezout's Theorem (\ref{Bezout}), two projective conics with no common irreducible components may meet at up to $4$ distinct points.

Therefore, in order to construct a ``good'' incidence structure from conics, i.e. a set of points and a set of blocks such that two points have at most one block in common, we use two ideas.
\begin{enumerate}[(1)]
\item First, we consider particular sets of affine conics such that the projective closures of any two of them intersect twice at infinity. Such curves meet at most twice in the affine plane. This is the point of the present section.
\item Second, we blow--up the rational points of the affine plane and consider the strict transforms of conics on this blown--up plane. From Proposition \ref{summarize} these strict transforms meet less frequently and provide an incidence structure which turns out satisfy our expectations.
This is the point of Section \ref{BU}.
\end{enumerate}

\subsection{Three sets of affine conics}

We describe three sets of affine conics respectively denoted by $\C_1 (q), \C_2 (q)$ and $\C_3 (q)$. They are constructed from $3$--dimensional linear system of conics having prescribed points or tangents at infinity.

\begin{defn}[The set $\C_1 (q)$]\label{defC1}
Let $\Gamma_1$ be the linear system of conics containing $P_{\infty}$ and tangent to $L_{\infty}$ at $P_{\infty}$ (see Notation \ref{Linfty}).
We define the set $\C_1 (q)$ to be the set affine conics defined over $\F_q$ whose projective closure is a smooth element of $\Gamma_1$.
In affine geometry over the reals, such conics would be a family of parabolas.
\end{defn}

\begin{defn}[The set $\C_2 (q)$]\label{defC2}
Let $\Gamma_2$ be the linear system of projective conics containing the points $P_{\infty}$ and $Q_{\infty}$ (see Notation \ref{Linfty}).
We define the set $\C_2 (q)$ to be the set of affine conics defined over $\F_q$ whose projective closure is a smooth element of $\Gamma_2$.
In affine geometry over the reals, such conics would be a family of hyperbolas.
\end{defn}

\begin{defn}[The set $\C_3 (q)$]\label{defC3}
Let $\Gamma_3$ be the linear system of projective conics containing the pair $(R_{\infty},\overline{R}_{\infty})$ (see Notation \ref{Linfty}). 
We define $\C_3 (q)$ to be the set of affine conics defined over $\F_q$ whose projective closure is a smooth element of $\Gamma_3$. 
In affine geometry over the reals, such conics would be a family of ellipses.
\end{defn}

\begin{rem}\label{AutC3}
  Even if the set of curves $\C_3 (q)$ depends on the choice of $\alpha$, given two choices $\alpha, \alpha'$ of elements of $\F_{q^2}\setminus \F_q$, Lemma \ref{Aut} (\ref{Aut1}) asserts the existence of $\sigma \in \PGL (3, \F_q)$ sending $\{(\alpha : 1 :0), (\alpha^q : 1 : 0)\}$ onto $\{(\alpha' : 1 : 0 ), ({\alpha '}^q : 1 : 0)\}$. Thus, $\C_3 (q)$ is unique up to isomorphism and changing the choice of $\alpha$ does not change the isomorphism class of the incidence structure or that of the LDPC code.
\end{rem}

\begin{figure}[h]
  \centering
  \begin{tabular}{cp{2cm}c}
    \includegraphics[scale=0.5]{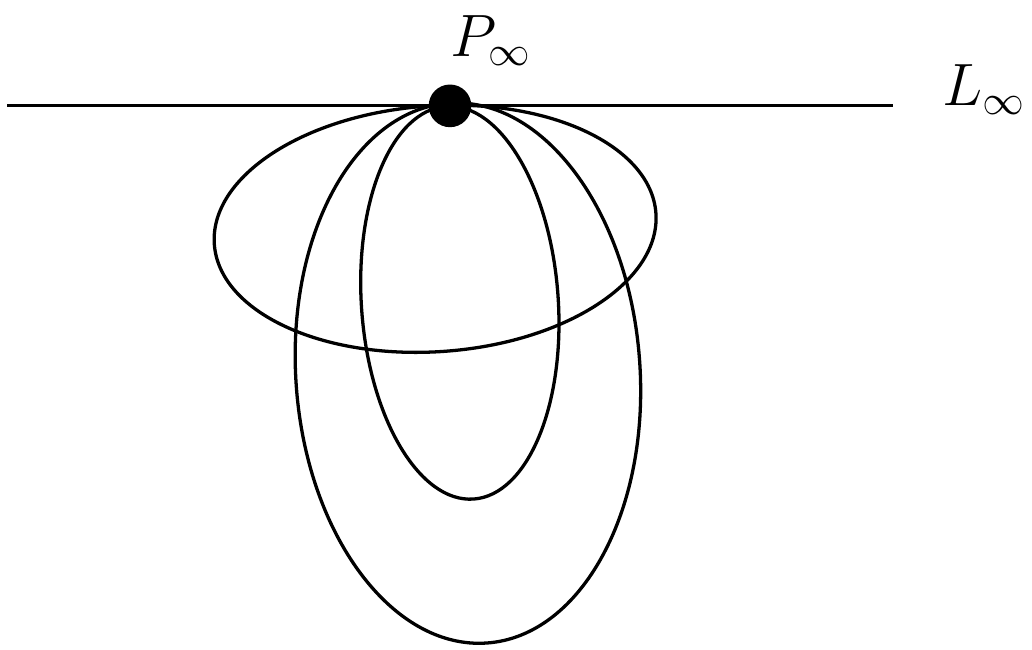}  
& & \includegraphics[scale=0.5]{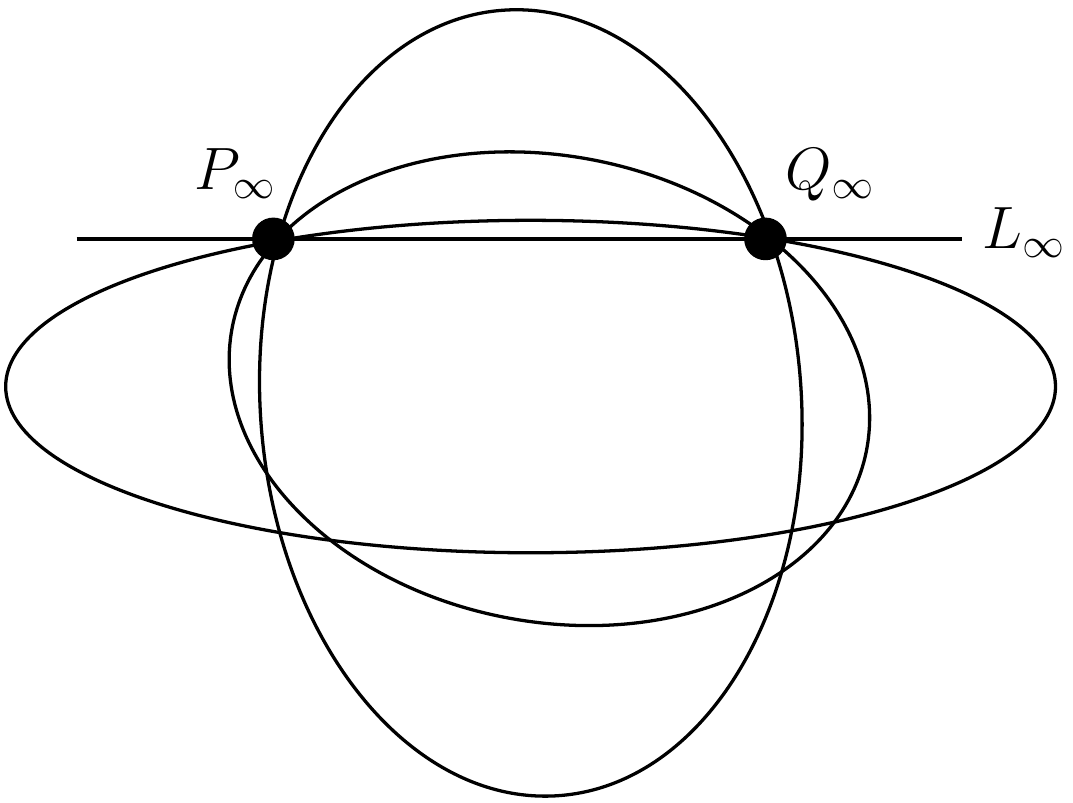}
  \end{tabular}  
  \caption{The linear systems $\Gamma_1$ (on the left) and $\Gamma_2$ (on the right)}
\end{figure}

\begin{figure}[h]
 \centering
  \begin{tabular}{cp{2cm}c}
  \includegraphics[scale=0.5]{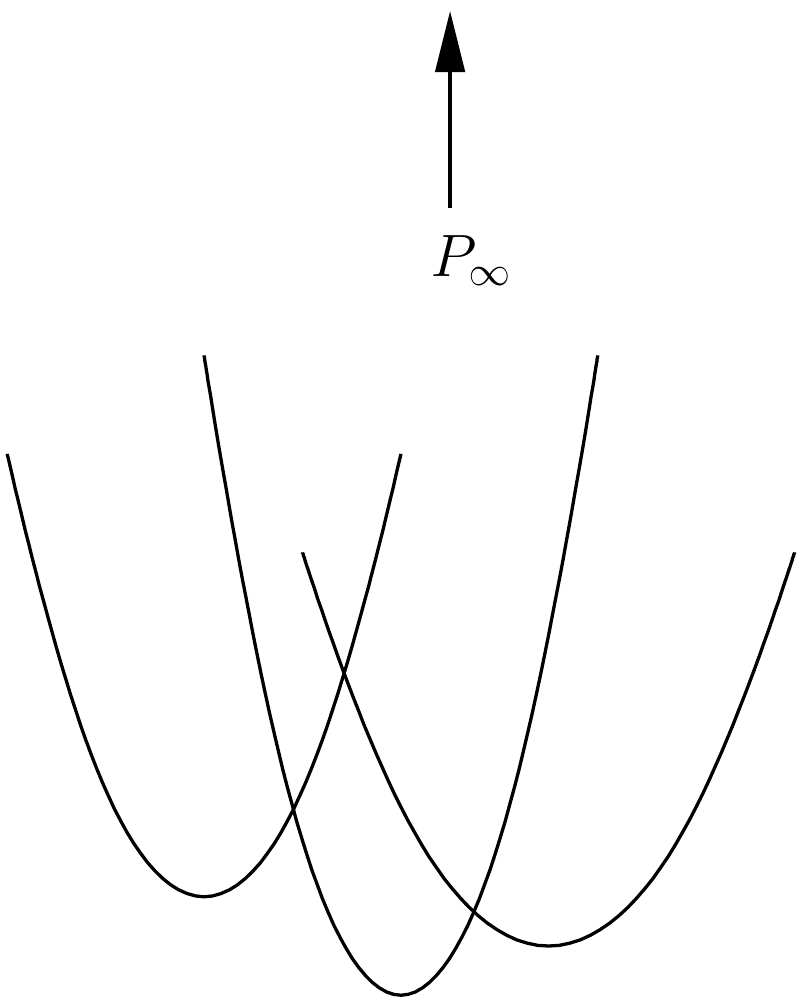} 
& &  \includegraphics[scale=0.5]{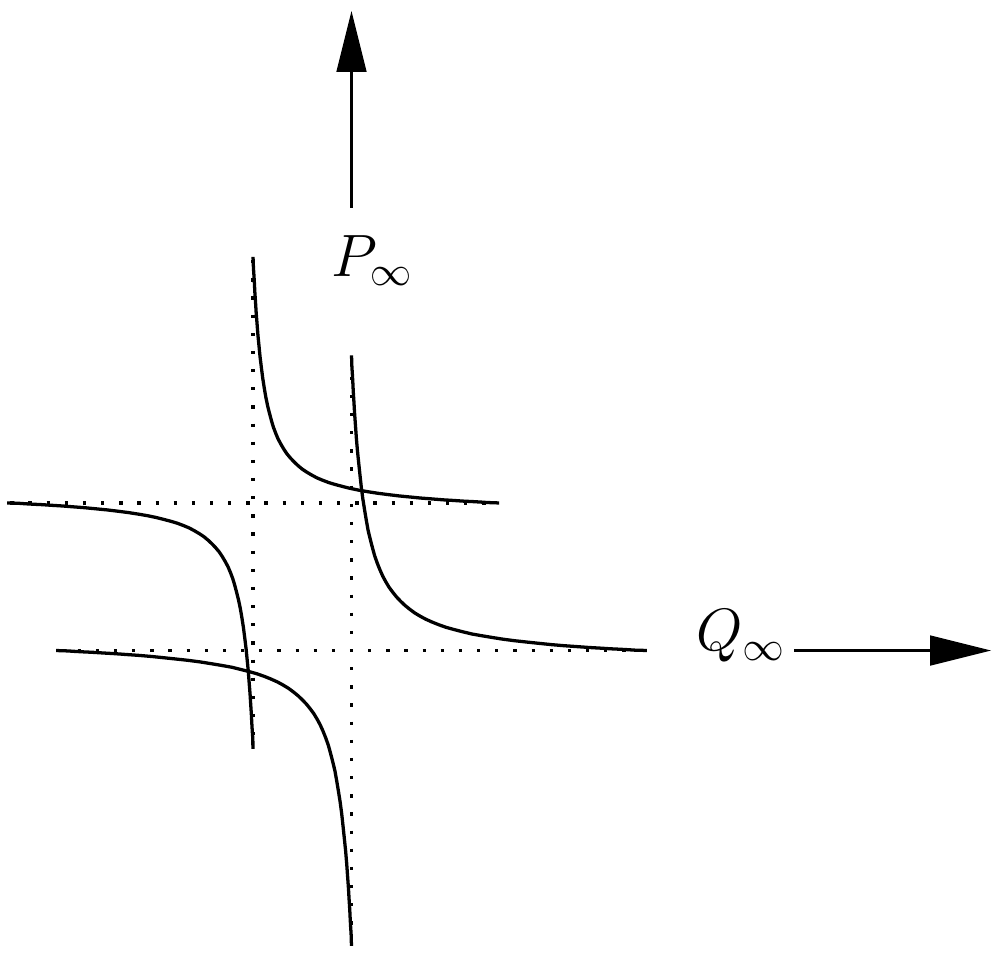}
  \end{tabular}  
  \caption{The sets $\C_1 (q)$ (on the left) and $\C_2 (q)$ (on the right)}
\end{figure}


\begin{rem}\label{caution}
In the linear systems $\Gamma_1, \Gamma_2$ and $\Gamma_3$, one finds some reducible conics obtained by the union of the line at infinity together with any other line. The trace of such curves in $\A^2$ is a line and hence a smooth irreducible curve but is not an affine conic. Such elements are not elements of the sets $\C_1 (q), \C_2 (q)$ or $\C_3 (q)$.
\end{rem}

\subsection{Explicit equations}

The following Lemmas give an explicit descriptions of the elements of $\C_1 (q), \C_2 (q)$ and $\C_3 (q)$.  

\begin{prop}[Equation of an element of $\C_1 (q)$]\label{eqC1}
  An affine curve is in $\C_1 (q)$ if and only if it has an equation of the form
$$
y=ax^2+bx+c,\quad \textrm{with}\quad a\in \F_q \setminus \{0\}\ \textrm{and}\ b,c \in \F_q.
$$
\end{prop}

\begin{proof}
Let $C \in \C_1 (q)$, let $\overline{C}$ be its projective closure and $F(X,Y,Z)=\lambda_1 X^2+ \lambda_2 XY+ \lambda_3 XZ+\lambda_4 Y^2+\lambda_5 YZ+\lambda_6 Z^2$ be a defining polynomial of $\overline{C}$.
The  condition $P_{\infty}=(0:1:0) \in \overline{C}$ entails $\lambda_4=0$.
For the tangency condition, consider the affine chart $\{Y\neq 0\}$.
In this chart, we get a non homogeneous equation $g(x,z)=\lambda_1 x^2+ \lambda_2 x+ \lambda_3 xz+\lambda_5 z+\lambda_6 z^2=0$. The point $P_{\infty}$ has coordinates $(0,0)$ in this chart.
From Definition \ref{Tangent}, being tangent at $P_{\infty}$ to $L_{\infty}$ (which has equation $z=0$ in this chart), means that $\frac{\partial g}{\partial x}(0,0)=0$ and hence $\lambda_2=0$. 
In addition, if $\lambda_5$ was zero, then $\frac{\partial g}{\partial z}$ would also vanish at $P_{\infty}$ and $\overline{C}$ would be singular. Thus, $\lambda_5\neq 0$ and can be set to $-1$ without loss of generality

In the affine chart $\{Z\neq 0\}$, the conic has an affine equation of the form $f(x,y)=\lambda_1 x^2+ \lambda_3 x+\lambda_6-y=0$.
Moreover, $\lambda_1$ must be nonzero or the corresponding affine curve would be a line and not a conic.
There remains to show that under these conditions $\overline{C}$ is always smooth. It is obviously smooth at $P_{\infty}$ (that was the reason why we set $\lambda_5 \neq 0$). It is also smooth in the affine chart $\{Z \neq 0\}$ since the partial derivative $\frac{\partial f}{\partial y}\equiv 1$ and hence never vanishes in this chart.
\end{proof}

\begin{prop}[Equation of an element of $\C_2 (q)$]\label{eqC2}
  An affine curve is in $\C_2 (q)$ if and only if it has an equation of the form
$$
xy=ax+by+c,\quad \textrm{with}\quad (a,b,c)\in \F_q^{3} \ \textrm{and}\ c\neq -ab.
$$
\end{prop}

\begin{proof}[Proof of Proposition \ref{eqC2}]
Let $C\in \C_2(q)$ and $\overline{C}$ and $F(X,Y,Z)$ be as in the proof of Proposition \ref{eqC1}.
The conditions at $P_{\infty}$ and $Q_{\infty} \in \overline{C}$  entail respectively $\lambda_4=0$ and $\lambda_1=0$.
This yields an affine equation for $C$ of the form $f(x,y)=\lambda_2 xy+ \lambda_3 x+\lambda_5 y+\lambda_6$ $=0$. Since $C$ is a conic and not a line, $\lambda_2\neq 0$ and can bet set to $-1$ without loss of generality. 
Then $\frac{\partial f}{\partial x}=\lambda_3-y$ and $\frac{\partial f}{\partial y}=\lambda_5-x$. Thus, $C$ is singular if and only if the point $(\lambda_5,\lambda_3)$ is in $C$. One checks easily that this situation happens if and only if $\lambda_6 = -\lambda_3\lambda_5$.
Therefore if  $\lambda_6 \neq -\lambda_3\lambda_5$ then $C$ is smooth.
There remains to check that $\overline{C}$ is also smooth at $P_{\infty}$ and $Q_{\infty}$.
Since $\overline{C}$ has degree $2$ and meets $L_{\infty}$ at $P_{\infty}$ and $Q_{\infty}$, from B\'ezout's Theorem, the intersection multiplicities $m_{P_{\infty}}(\overline{C},L_{\infty})$ and $m_{Q_{\infty}}(\overline{C},L_{\infty})$ are both equal to $1$ and hence $\overline{C}$ cannot be singular at these points (see \S \ref{IntMult}).
\end{proof}

The equation of the elements of $\C_3 (q)$ depend on the choice of $\alpha$. From Remark \ref{AutC3}, there is no loss of generality to consider an arbitrary choice of $\alpha$, which is what we do in the two following lemmas.

\begin{prop}[Equation of an element of $\C_3 (q)$ in odd characteristic]\label{eqC3odd}
Assume that $q$ is odd. 
Let $\beta$ be a non-square element of $\F_q \setminus \{0\}$ and $\alpha\in \F_{q^2}\setminus \F_q$ such that $\alpha^2=\beta$.
For this choice of $\alpha$, an affine curve is in $\C_3 (q)$ if and only if it has an equation of the form
$$
x^2-\beta y^2=ax+by+c, \quad  with\quad c\neq \frac{b^2}{4\beta}-\frac{a^2}{4} \cdot
$$
\end{prop}

\begin{proof}
Let $C\in \C_3(q)$ and $\overline{C}$ and $F(X,Y,Z)$ be as in the proof of Proposition \ref{eqC1}.
The vanishing conditions at $R_{\infty}$ and $\overline{R}_{\infty} \in \overline{C}$ entail
$\lambda_1 \alpha^2 + \lambda_2 \alpha +\lambda_4 = 0$.
Since $T^2-\beta$ is the minimal polynomial of $\alpha$, we have $\lambda_2=0$ and $\lambda_4=-\beta \lambda_1$. If $\lambda_1=0$, then $C$ would have an affine equation of the form $\lambda_3 x +\lambda_5 y + \lambda_6=0$ and hence would not be a conic. Therefore, $\lambda_1$ is nonzero and can be set to $-1$ without loss of generality.
As in the proof of Lemma \ref{eqC2}, an argument based on B\'ezout's Theorem asserts that on these conditions $\overline{C}$ is smooth at $R_{\infty}$ and $\overline{R}_{\infty}$. There remains to find under which additional conditions it is smooth in the affine chart $\{Z\neq 0\}$. In this chart, the curve has an equation of the form $f(x,y)= \lambda_3 x + \lambda_5 y + \lambda_6 -(x^2-\beta y^2)$. 
A computation of the partial derivatives of $f$ entails that $C$ is singular if and only if $f(\frac{\lambda_3}{2}, -\frac{\lambda_5}{2\beta})=0$. This leads to the assertion that $C$ is smooth provided $\frac{\lambda_3^2}{4}-\frac{\lambda_5^2}{4\beta}+\lambda_6 \neq 0$.
\end{proof}

\begin{prop}[Equation of an element of $\C_3 (q)$ in even characteristic]\label{eqC3even}
 Assume that $q$ is even. 
Let $\beta$ be an element of $\F_q\setminus \{0\}$ such that ${\rm Tr}_{\F_{q^2}/\F_q}(\beta)\neq 0$ and $\alpha\in \F_{q^2}\setminus \F_q$ such that $\alpha^2+\alpha +\beta =0$.
For this choice of $\alpha$, an affine curve is in $\C_3 (q)$ if and only if it has an equation of the form
$$
x^2+xy +\beta y^2=ax+by+c, \quad  with\quad c \neq a^2+b^2+ab.
$$ 
\end{prop}

\begin{proof}
  The proof is similar as that of Proposition \ref{eqC3odd}, the conditions $R_{\infty}$ and $\overline{R}_{\infty}$ entail $\lambda_1=\lambda_2$ and $\lambda_3=\beta \lambda_1$. Moreover, $\lambda_1 \neq 0$ and can be set to $1$ without loss of generality. Using B\'ezout's Theorem, one asserts that $\overline{C}$ is smooth at the points at infinity.

In the affine chart $\{Z\neq 0\}$, the curve $C$ has an equation of the form
$f(x,y)=\lambda_3 x + \lambda_5 y + \lambda_6 +(x^2+x+\beta y^2)$ and computations on the partial derivatives entail that this curve is smooth provided $\lambda_3^2+\lambda_5^2+\lambda_3\lambda_5+\lambda_6 \neq 0$.
\end{proof}

\subsection{Counting number of elements}

\begin{prop}[Cardinal of the $\C_i (q)$'s]\label{card1}
  The sets $\C_1 (q), \C_2 (q)$ and $\C_3 (q)$ have $q^3-q^2$ elements.
\end{prop}

\begin{proof}
It is a straightforward consequence of Lemmas \ref{eqC1} to \ref{eqC3even}.
\end{proof}

\begin{prop}[Number of points of an element of $\C_i (q)$]\label{card2}
  Any element of $\C_1 (q)$ (resp. $\C_2 (q)$, resp. $\C_3 (q)$) has exactly $q$ (resp. $q-1$, resp $q+1$) rational points in $\A^2$.
\end{prop}

\begin{proof}
By definition, for all $C\in \C_1 (q)$ (resp. $\C_2 (q), \C_3 (q)$) the projective closure $\overline{C}$ is smooth.
Thus, from Proposition \ref{Classification} (\ref{irred}), $\overline{C}$ has $q+1$ rational points.
Moreover, it has one (resp. two, resp. zero) prescribed rational point at infinity, this yields the result.
\end{proof}

\subsection{Incidence relations}\label{subsecinc}

\begin{lem}[Incidence structures given by the $\C_i (q)$'s]\label{incid}
  Let $C,D$ be a pair of distinct elements of $\C_1 (q)$ (resp. of $\C_2 (q)$, resp. of $\C_3 (q)$). Then, these two curves meet at $0$ or $1$ or $2$ rational points of $\A^2$.
Moreover, if they have a common tangent at a rational point $P$ of $\A^2$, then they do not meet at another point of $\A^2$. 
\end{lem}

\begin{proof}
By definition, any two conics of $\C_1 (q)$ (resp. $\C_2 (q)$, resp. $\C_3 (q)$) meet at least twice (counted with multiplicities) at infinity. 
This claim together with B\'ezout's Theorem yield the result. 
\end{proof}

\subsection{Affine automorphisms}
To conclude the present section, we focus on automorphisms of $\A^2$ preserving $\C_1 (q)$ (resp. $\C_2 (q)$, resp. $\C_3 (q)$). Basically they are the projective automorphisms preserving the pair $(P_{\infty}, L_{\infty})$ (resp. $(P_{\infty}, Q_{\infty})$, resp. $R_{\infty}, \overline{R}_{\infty}$).
We have the following lemma.

\begin{lem}\label{AutBis}
  The group of automorphisms of $\A^2$ preserving $\C_1 (q)$ (resp. $\C_2 (q)$, resp. $\C_3 (q)$) acts transitively on the pairs $(P,L)$ such that $P$ is a rational point of $\A^2$ and $L$ is a non vertical line containing $P$  (resp. a neither vertical nor horizontal line containing $P$, resp. a line containing $P$).
\end{lem}

\begin{proof}
  It is a consequence of Lemma \ref{Aut}.
\end{proof}

\section{The blown--up plane}\label{BU}

As said in Lemma \ref{incid}, two conics of $\C_i (q)$ may meet at two distinct points, therefore the incidence structure given by the rational points of $\A^2$ together with the elements of $\C_1 (q)$ (resp. $C_2, \C_3 (q)$) does not satisfies the conditions expected in \S \ref{aims}. This is the reason why, we introduce the surface $\B$.

\begin{defn}[The surface $\B$]\label{BBu}
Let $P_1, \ldots , P_{q^2}$ be the rational points of the affine plane.  
The surface $\B$ is the surface obtained from $\A^2$ by blowing up all the points $P_1, \ldots , P_{q^2}$.
The corresponding exceptional divisors are denoted by $E_{P_1}, \ldots , E_{P_{q^2}}$ and we denote by $\mathcal{E}$ the set $\mathcal{E}:=\{E_{P_1}, \ldots , E_{P_{q^2}}\}$.
\end{defn}

\begin{rem}\label{disj}
  Two distinct exceptional divisors on $\B$ are disjoint.
\end{rem}

The rational points of $\B$ can be interpreted in terms of \emph{flags}. This is the purpose of the following definition.

\begin{defn}[Flags] We call a \emph{flag} on $\A^2$ a pair $(P,L)$, where $P$ is a rational point of $\A^2$, $L$ is a line defined over $\F_q$ and $P\in L$.
\end{defn}

\begin{defn}[Incidence flag/curve]
 A plane curve $C$ and a flag $(P,L)$ are said to be \emph{incident }if $P\in C$ and $L$ is a tangent to $C$ at $P$. 
\end{defn}

\begin{lem}\label{flags}
  The rational points of $\B$ are in one-to-one correspondence with the flags of $\A^2$. In particular, $\B$ has $q^2(q+1)$ rational points.
\end{lem}

\begin{proof}
  It is a straightforward consequence of Proposition \ref{summarize} (\ref{sum1}).
\end{proof}

The following theorem summaries most of the basic elements needed in the study of the further described incidence structures.

\begin{thm}\label{truc}
  Let $(P,L)$ be a flag in the affine plane, then
  \begin{enumerate}[(i)]
  \item\label{c1t} $(P,L)$ is incident with $q-1$ elements of $\C_1 (q)$ if $L$ is not a vertical (see Definition \ref{vertical}), else it is tangent to none of them;
  \item\label{c2t} $(P,L)$ is incident with $q-1$ elements of $\C_2 (q)$ if $L$ is neither vertical nor horizontal (see Definition \ref{vertical}), else it is tangent to none of them;
  \item\label{c3t}  $(P,L)$ is always incident with $q-1$ elements of $\C_3 (q)$.
  \end{enumerate}
\end{thm}

\begin{proof}
 \textit{Step 1.} First suppose that $L$ is vertical, i.e. equal to $(PP_{\infty})$ (see Notation \ref{notdtes}) and assume that there exists $C\in \C_1 (q)$ which is tangent to $L=(PP_{\infty})$ at $P$. Let $\overline{C}$ be the projective closure of $C$. By definition, $\overline{C}$ contains $P_{\infty}$.
Then, $m_P(\overline{C}, L)\geq 2$ and $m_{P_{\infty}}(\overline{C},L)\geq 1$, which contradicts B\'ezout's Theorem since $\deg(\overline{C}).\deg(L)=2$.

\smallbreak

\noindent \textit{Step 2.} Suppose that $L$ is non-vertical. Then, the set of conics containing $P_{\infty}$ and $P$ and which are respectively tangent to $L_{\infty}$ and $L$ at these points is the linear system $\Lambda_2(P, P_{\infty}, L, L_{\infty})$ (see Lemma \ref{contrick2} and Remark \ref{cardi}).
It has dimension $1$, thus has $q+1$ elements defined over $\F_q$ and from Lemma \ref{contrick2}, they  are all smooth but two of them. 

\smallbreak

One proves (\ref{c2t}) and (\ref{c3t}), by the very same manner using Lemma \ref{contrick} instead of \ref{contrick2}.
\end{proof}

\section{The new incidence structures}\label{secinc}

In this section we describe three incidence structures obtained from the surface $\B$ and the $\C_i (q)$'s.

\subsection{Description}

\begin{defn}\label{defwCi}
  The sets ${\wC}_1 (q), \wC_2 (q)$ and $\wC_3 (q)$ are the respective sets of strict transforms of the elements of $\C_1 (q), \C_2 (q)$ and $\C_3 (q)$ on $\B$.
\end{defn}

\noindent Recall that $\E$ denotes the set of all exceptional divisors on $\B$.

\begin{defn}[Incidence structure $\I_1 (q)$]\label{I1}
  We denote by $\I_1 (q)$ the incidence structure whose set of points $\Pc_1 (q)$ is the set of rational points of $\B$ corresponding to the flags $(P,L)$ of $\A^2$ such that $L$ is not a vertical and whose set of blocks $\Bc_1 (q)$ is $\wC_1 (q) \cup \E$.
\end{defn}

\begin{defn}[Incidence structure $\I_2 (q)$]\label{I2}
  We denote by $\I_2 (q)$ the incidence structure whose set of points $\Pc_2 (q)$ is the set of rational points of $\B$ corresponding to the flags $(P,L)$ of $\A^2$ such that $L$ is neither vertical nor horizontal and whose set of blocks $\Bc_2 (q)$ is $\wC_2 (q) \cup \E$.
\end{defn}

\begin{defn}[Incidence structure $\I_3 (q)$]\label{I3}
  We denote by $\I_3 (q)$ the incidence structure whose set of points $\Pc_3 (q)$ is the set of all rational points of $\B$ and whose set of blocks $\Bc_3 (q)$ is $\wC_3 (q) \cup \E$.
\end{defn}

\begin{rem}[Why adding the exceptional divisors?]\label{whyEx}
In the above-described incidence structures, one can wonder why we chose to add the exceptional divisors in the block sets. Actually the incidence structures would have been regular without these blocks. However, by adding a negligible number of blocks ($q^2$ additional blocks in a set containing already $q^3-q^2$ blocks), one gets codes with a twice larger minimum distance, see Remark \ref{WhyEx2}. 
\end{rem}

\subsection{Basic properties of the incidence structures}

\begin{nota}\label{EP}
  \begin{enumerate}[(1)]
  \item  In what follows, any element of $\E$, (i.e. any exceptional divisor on $\B$) is denoted by $E_P$, where $P$ is the corresponding blown--up point (i.e. the image of $E_P$ by the canonical map $\B \rightarrow \A^2$).
 \item Using Lemma \ref{flags}, any rational point of $\B$ is represented as a flag $(P,L)$ on $\A^2$.
We allow ourselves the notation ``$(P,L)\in \B$''.
In particular, we have $(P,L) \in E_P$.
 \item From now on, an element of $\wC_i (q)$ is denoted by $\widetilde{C}$, where $C$ is the affine conic whose strict transform is $\widetilde{C}$.
  \end{enumerate}
\end{nota}
\medbreak

\begin{thm}\label{main}
The incidence structures $\I_1 (q), \I_2 (q)$ and $\I_3 (q)$ satisfy
\begin{enumerate}[(i)]
\item\label{m1} $\sharp \Bc_1 (q) = q^3$ and $\sharp \Pc_1 (q)=q^3$ ;
\item\label{m2} $\sharp \Bc_2 (q) = q^3$ and $\sharp \Pc_2 (q)=q^2(q-1)$;
\item\label{m3} $\sharp \Bc_3 (q) = q^3$ and $\sharp \Pc_3 (q)=q^2(q+1)$.
\item\label{m4} any point of $\I_1 (q)$ (resp. $\I_2 (q)$, resp. $\I_3 (q)$) is incident with exactly $q$ blocks;
\item\label{m5} any block of $\Bc_1 (q)$ (resp. $\Bc_2 (q)$, resp. $\Bc_3 (q)$) is incident with exactly $q$ (resp. $q-1$, resp. $q+1$) points.
\item\label{m6} any two distinct points of $\I_1 (q)$ (resp. $\I_2 (q)$, resp. $\I_3 (q)$) are incident with at most $1$ common block.
\end{enumerate}
\end{thm}

\begin{proof}
  From Proposition \ref{card1}, the sets $\C_1 (q), \C_2 (q),\C_3 (q)$ and hence the sets $\wC_1 (q), \wC_2 (q),$ $\wC_3 (q)$ have cardinal $q^3-q^2$. Since $\E$ has cardinal $q^2$, this proves that the block sets $\Bc_i (q)'s$ have cardinal $q^3$.
From Lemma \ref{flags}, the surface $\B$ has $q^2(q+1)$ rational points.
To construct $\Pc_1 (q)$ (resp. $\Pc_2 (q)$) we remove $1$ (resp. $2$) rational point(s) per exceptional divisor, corresponding to the vertical direction (resp. vertical and horizontal directions).
Consequently, these sets have respectively $q^3$ and $q^2(q-1)$ elements.
To construct $\Pc_3 (q)$ we take all the rational points of $\B$, thus this set has cardinal $q^2(q+1)$.
This proves (\ref{m1}), (\ref{m2}) and (\ref{m3}).

\medbreak

A point of $\I_1 (q)$ (resp. $\I_2 (q)$, resp. $\I_3 (q)$) corresponds to a flag $(P,L)$ of $\A^2$ such that $L$ is a non-vertical line (resp. is a neither vertical nor horizontal line, resp. is a line) of $\A^2$.
From Theorem \ref{truc}(\ref{c1t}) (resp. (\ref{c2t}), resp. (\ref{c3t})), such a flag is incident with exactly $q-1$ conics of $\C_i (q)$ and hence the corresponding point of $\B$ is in $q-1$ elements of $\wC_i (q)$.
In addition, this point also lies in the exceptional divisor $E_P$.
Therefore, the point $(P,L)\in \B$ is incident with $q$ blocks in $\Bc_1 (q)$ (resp. $\Bc_2 (q)$, resp. $\Bc_3 (q)$). This proves (\ref{m4}).

\medbreak

The number of points incident to a block of the form $\wC_i (q)$ is a straightforward consequence of Proposition \ref{card2} together with Lemma \ref{same}.
For the blocks in $\E$, first recall that an exceptional divisor is isomorphic to a projective line and hence has $q+1$ rational points. Moreover, to construct  $\Pc_1 (q)$ (resp. $\Pc_2 (q)$) we take all the flags but those of the form $(P,L)$ where $L$ is vertical (resp. either vertical or horizontal). Thus, we remove $1$ (resp. $2$) rational point to each exceptional divisor. Consequently, any block in $\Bc_1 (q)$ (resp. $\Bc_2 (q)$, resp. $\Bc_3 (q)$) from $\E$ is incident with $q$ (resp. $q-1$, resp. $q+1$) points in $\Pc_1 (q)$ (resp. $\Pc_2 (q)$, resp. $\Pc_3 (q)$). This proves (\ref{m5}).

\medbreak

Let $i\in \{1,2,3\}$.
  Let $(P,L),(P',L')$ be two distinct points of $\Pc_i (q)$. First, assume that they both lie in the same exceptional divisor $E_P$ of $\B$, i.e. $P=P'$.
Then, $E_P$ is the only block containing both of them.
Indeed, no other exceptional divisor contains them from Remark \ref{disj}.
Moreover, by definition, elements of $\C_i (q)$ are smooth. Therefore, from Proposition \ref{summarize} (\ref{sum1}), a curve $\widetilde{C}\in \wC_i (q)$ meets $E$ at at most one point and hence cannot contain both points $(P,L)$ and $(P,L')$.
Now, suppose that the points $(P,L),(P',L')$ lie in distinct exceptional divisors (i.e. $P\neq P'$). Then, there is at most one element of $\wC_i (q)$ containing both of them.
Indeed, assume that there exists two distinct such curves $\widetilde{C},\widetilde{D}\subset \wC_i (q)$ both containing the points $(P,L)$ and $(P',L')$.
Then, from Proposition \ref{summarize} (\ref{IncMult}), the curves $C, D$ both contain the points $P, P'$ and meet with multiplicity $\geq 2$ at both them.
Therefore, these curves meet at $4$ points of $\A^2$ counted with their multiplicities, which contradicts Lemma \ref{incid}. This proves (\ref{m6}).
\end{proof}

\subsection{Girths}\label{secgirth}

\begin{thm}[Girth of the incidence graph]\label{girth}
  Let $\gamma(i,q)$ be the girth of the incidence graph of the incidence structure $\I_i (q)$. We have
$$
\begin{array}{ccl}
  \gamma (1,q) & = &
\left\{  \begin{array}{cccl}
    6 & if &  q & even \\
    8 & if & q & odd 
  \end{array} \right. \\
 & & \\
 \gamma (2,q) & = & \left\{  \begin{array}{cccl}
    8 & if & q & even \\
    6 & if & q & odd 
  \end{array} \right. \\
 & & \\
 \gamma (3,q) & = & \left\{  \begin{array}{cccl}
    8 & if & q & even \\
    6 & if & q & odd 
  \end{array} \right. .
\end{array}
$$
\end{thm}

\begin{rem}\label{evengirth}
  The incidence graph of an incidence structure is a bipartite graph. Therefore, its cycles have even length. Thus, the girth of such graphs is even.
\end{rem}

To prove Theorem \ref{girth}, we need Lemma \ref{6cycl}, Lemma \ref{8cycl}, Proposition \ref{MultiInc} and Lemma \ref{Geo8}.

\begin{defn}[$(C3)$ configurations]\label{C3}
  A triple $\{C_a, C_b , C_c\}$ of distinct plane affine conics is said to be a $(C3)$ configuration if
 \begin{enumerate}[(i)]
 \item\label{C31}
any two of the conics meet at only one point in $\A^2$ and are tangent at this point;
 \item\label{C32} in $\A^2$ we have $C_a \cap C_b \cap C_c = \emptyset$.
 \end{enumerate}
See Figure \ref{Fig_C3} for an illustration.
\end{defn}

\begin{figure}[!h]
  \centering
    \includegraphics[scale=0.5]{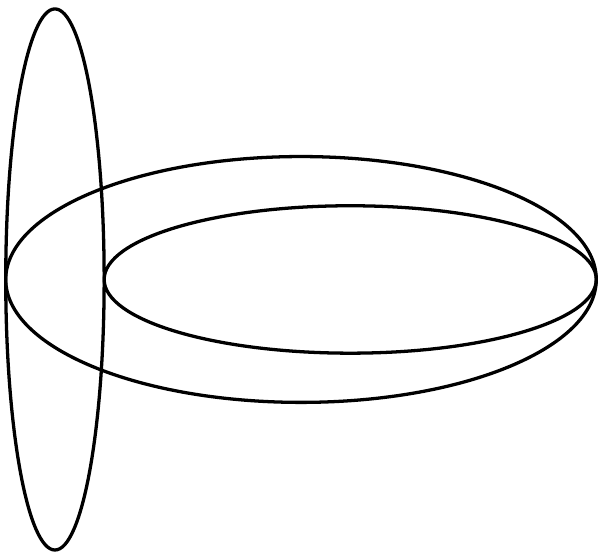}
  \caption{A triple of conics in $(C3)$ configuration}
  \label{Fig_C3}
\end{figure}

\begin{lem}[The geometric structure of $6$--cycles]\label{6cycl}
  The incidence graph of $\I_1 (q)$ (resp. $\I_2 (q), \I_3 (q)$) has cycles of length $6$ if and only if three elements of $\C_1 (q)$ (resp. $\C_2 (q), \C_3 (q)$) are in $(C3)$ configuration.
\end{lem}

\begin{proof}

Obviously, a triple of conics in $(C3)$ configuration yields a cycle of length $6$ in the incidence graph.

Conversely, assume there is a cycle of length $6$ in the incidence graph of $\I_i (q)$.
Then, there are three blocks $B_a, B_b, B_c \in \Bc_i (q)$ together with three points  $(P_{bc},L_{bc}), (P_{ac}, L_{ac}),$ $(P_{ab}, L_{ab}) \in \Pc_i (q)$ such that $(P_{bc}, L_{bc})$ (resp. $(P_{ac}, L_{ac})$, resp. $(P_{ab}, L_{ab})$) is incident with $B_b$ and $B_c$ (resp. $B_a$ and $B_c$, resp. $B_a$ and $B_b$).
From Remark \ref{disj}, two exceptional divisors (blocks in $\E$) have no common point in $\Pc_i (q)$.
Therefore, at least two of the $B_i$'s, say $B_a, B_b$ are not in $\E$ and hence are of the form $\widetilde{C}_a, \widetilde{C}_b$ with $C_a, C_b\in \C_i (q)$.
Let us prove that $B_c$ cannot be an exceptional divisor.
If it was, since points $(P_{bc}, L_{bc})$ and $(P_{ac},L_{ac})$ are both incident with $B_c$, we would have $P_{bc}=P_{ac}$.
In addition, since $B_a=\widetilde{C}_a$ and $B_b=\widetilde{C}_b$ are both incident with $(P_{ab}, L_{ab})$, from Proposition \ref{summarize} (\ref{IncMult}), the plane curves $C_a$ and $C_b$ meet at least ``twice'' at $P_{ab}$ (i.e. have a common tangent line $L_{ab}$ at this point). 
They also meet at $P_{bc}=P_{ac}$, which contradicts Lemma \ref{incid}.

Therefore, $B_a, B_b, B_c$ are curves of the form $\widetilde{C}_a, \widetilde{C}_b$ and $\widetilde{C}_c$ and $C_a, C_b$ (resp. $C_a, C_c$, resp. $C_b, C_c$) meet ``twice'' at $P_{ab}$ (resp. $P_{ac}$, resp. $P_{bc}$) i.e. have a common tangent line $L_{ab}$ (resp $L_{ac}$, resp. $L_{bc}$) at this point. 
This means that the affine conics $C_a, C_b, C_c$ are in $(C3)$ configuration. 
\end{proof}

\begin{lem}[Criterion for the non-existence of $(C3)$ configurations]\label{8cycl}
Let $i\in \{1,2,3\}$.
  Let $(P,L)$ be an element of $\Pc_i (q)$ (i.e. a flag such that $L$ is a non-vertical line, resp. a neither vertical nor horizontal line, resp. a line).
If for all $C\in\C_i (q)$ such that $P\notin C$, there exists at most one conic $D\in \C_i (q)$ incident with $(P,L)$ and such that $C,D$ are incident with a common flag $(P', L')\in \Pc_i (q)$, then no three elements of $\C_i (q)$ are in $(C3)$ configuration.
\end{lem}

\begin{proof}
Assume that there exists a triple $C_a, C_b, C_c \in \C_i (q)$ of conics in $(C3)$ configuration.
Let $(P_{ab},L_{ab})$ be the flag such that $P\in C_a \cap C_b$ and $L_{ab}=T_{P_{ab}} C_a=T_{P_{ab}} C_b$.
Applying a suitable automorphism of the plane (see Lemma \ref{AutBis}), one can assume that $(P_{ab}, L_{ab})=(P,L)$.
By definition of $(C3)$ configurations, $P\notin C_c$ and
there are two distinct curves (namely $C_a$ and $C_b$) incident with $(P,L)$ and having a common flag with $C_c$. This contradicts the assumption of the statement.
\end{proof}

\begin{prop}\label{MultiInc}
Let $(P,L)$ be a flag in $\Pc_i (q)$ and $C$ be an element of $\C_i (q)$ such that $P\notin C$. Denote by $\kappa_i (q,P,L,C)$ the maximum number of elements $C'\in C_i (q)$ such that
 $C'$ is incident with $(P,L)$ and
 $C,C'$ are incident with a common flag $(P', L')\in \Pc_i (q)$.
$$
\begin{array}{lllll}
  \textrm{If}\ q\ \textrm{is even,} &\ & \kappa_1 (q, P, L, C) = q-2 & \kappa_2 (q,P,L,C) = 1 & \kappa_3 (q,P,L,C) = 1 \\
  \textrm{If}\ q\ \textrm{is odd,} & & \kappa_1 (q,P,L,C) = 1 & \kappa_2 (q,P,L,C) \leq 2 & \kappa_3 (q,P,L,C) \leq 4.
\end{array}
$$
\end{prop}

\begin{rem}\label{junction}
  From Lemma \ref{8cycl}, if $\kappa (q, P, L, C)\leq 1$ for all $P,L,C$ such that $P\notin C$, then no $3$ elements of $\C_i (q)$ are in $(C3)$ configuration.
From Lemma \ref{6cycl}, such a condition also entails that the incidence graph of $\I_i (q)$ has no cycles of length $6$.
\end{rem}

\begin{proof}[Proof of Proposition \ref{MultiInc}]
\noindent \emph{Step 1. For $i=1$.} From Lemma \ref{AutBis}, without loss of generality, one can choose $(P,L)$ such that $P=(0,0)$ and $L=\{y=0\}$. A curve $C\in \C_1 (q)$ avoiding $P$ has an equation of the form (see Proposition \ref{eqC1}):
$$
(C):\quad  y=ax^2+bx+c,\ {\rm with}\ a\neq 0\ {\rm and}\ c\neq 0.
$$
A curve $C_t\in \C_1 (q)$ incident with $(P,L)$ has an equation of the form
$$
(C_t):\quad y=tx^2,\ {\rm where}\ t\neq 0.
$$
A point $P$ of intersection of $C$ and $C_t$ has coordinates satisfying both equations. Thus, its $x$--coordinate satisfies
\begin{equation}
  \label{poly1}
  (a-t)x^2+bx+c=0
\end{equation}
The curves $C$ and $C_t$ have a common tangent at $P$ if they meet with multiplicity $2$ at this point. It happens if equation (\ref{poly1}) has a double root.

\begin{itemize}
\item {\bf If $q$ is even} then (\ref{poly1}) has a double root if and only if $a\neq t$ and $b=0$. Therefore, if $b=0$, then $C$ shares a common flag with any curve $C_t$, with $t\neq 0$ and $t\neq a$. Else it does not share a common flag with any of the $C_t$. This yields $\kappa_1 (q,P,L,C)=q-2$. 

\item {\bf If $q$ is odd} then (\ref{poly1}) has a double root if and only if $a\neq t$ and its discriminant vanishes. Since $c\neq 0$, it discriminant
$\Delta (t)=b^2-4(a-t)c$ is a polynomial of degree $1$ in $t$ which vanishes for one value of $t$. This entails $\kappa_1 (q,P,L,C) = 1$.
\end{itemize}

\medbreak

\noindent \emph{Step 2. For $i=2$.} We choose $(P,L)$ with $P=(0,0)$ and $L=\{y=x\}$. A curve $C\in \C_2 (q)$ avoiding $P$ has an equation of the form (see Proposition \ref{eqC2}):
$$
(C): xy=ax+by+c \quad {\rm where}\ c\neq ab\ {\rm and}\ c\neq 0.
$$
A curve $C_t \in \C_2 (q)$ incident with $(P,L)$ has an equation of the form 
$$
(C_t):\quad xy=t(x+y)\quad {\rm where}\ t\neq 0.
$$
Notice that no point of $C_t$ has its $x$--coordinate equal to $t$.
Thus, from now on one can assume that $x\neq t$. The equation of $C_t$ can then be re-written as $y=\frac{xt}{x-t}$. Using this substitution in the equation of $C$, a quick computation gives
\begin{equation}
  \label{poly2} 
(t-a)x^2+(t(a-b)-c)x+ct=0 
\end{equation}

\begin{itemize}
\item {\bf If $q$ is even} then (\ref{poly2}) has a double root if and only if $a\neq t$ and $t(a+b)+c=0$. It happens for one value of $t$ if $a\neq b$ and does not happens if $a=b$.
This yields $\kappa_2 (q,P,L,C) = 1$.
\item {\bf If $q$ is odd} then (\ref{poly2}) has a double root if and only if $a\neq t$ and if its discriminant vanishes. Its discriminant
$\Delta (t)=(t(a-b)-c)^2-4ct(t-a)$ is a polynomial of degree $\leq 2$ in $t$ and hence vanishes for at most $2$ values of $t$. This yields $\kappa_2 (q,P,L,C) \leq 2$.
\end{itemize}

\medbreak

\noindent \emph{Step 3.a. For $i=3$ and $q$ even.} We choose $(P,L)$ as in Step 1. Using the same notations as in the previous steps we get
$$
(C):\quad x^2+xy+\beta y^2 = ax+by+c,\quad {\rm where},\ c\neq 0\ {\rm and}\ c\neq a^2+b^2+ab.
$$
and
$$
(C_t):\quad x^2+xy+\beta y^2 = tx, \quad {\rm where}\ t\neq 0.
$$

If $a=t$ and $b=0$, then the polynomial system has no solution since $c\neq 0$.
Else if $a=t$ and $b\neq 0$,
then after substitution we see that $C_t$ is tangent to $C$ at some point if and only if the polynomial
$
b^2x^2+b(c+bt)x+\beta c^2
$
has a double root. It happens only if $c=bt$.

If $a\neq t$, then
after substitutions, one sees that $C_t$ meets $C$ at a point of $\A^2$ with multiplicity $2$ if and only if the polynomial
$$
(b^2+(b+\beta(a+t))(a+t))y^2+(a+t)(c+bt)y+c^2+ct(a+t)
$$
has a double root. Since $a\neq t$, it happens only if $c=bt$.

Finally $C_t$ and $C$ meet with multiplicity $2$ if and only if $c=bt$. Since $c\neq 0$, this is possible for one value of $t$ when $b\neq 0$. This yields $\kappa_3 (q, P,L,C) = 1$.

\medbreak

\noindent \emph{Step 3.b. For $i=3$ and $q$ odd.} We choose $(P,L)$ as in Step 1, and get
$$
(C): \quad x^2-\beta y^2=ax+by+c,\ {\rm with}\ c\neq 0\ {\rm and}\ c\neq \frac{b^2}{4\beta}-\frac{a^2}{4},
$$
and 
$$
(C_t):\ \ x^2-\beta y^2=tx\ \ {\rm where}\ t\neq 0.
$$
The coordinates of a point at the intersection of $C_c$ and $C_t$ satisfy $(a-t)x+by+c=0$ which can be rewritten as $x=\frac{c}{a-t}-\frac{by}{a-t}$. Substituting this relation in the equation of $C_t$ and after a quick computation, we get
$$
(b^2-\beta^2)y^2+(tb(a-t)-2bc)y+c^2-c(a-t)=0
$$
The discriminant of this polynomial has degree $\leq 4$ in $t$. Thus, $\kappa_3 (q,P,L,C)\leq 4$.
\end{proof}

\begin{rem}\label{CyclC1even}
  In the above proof, Step 1 for $q$ even entails that for a fixed flag $(P,L)$, there are exactly $(q-1)^2$ curves $C\in \C_1 (q)$ avoiding $P$ and for which there exists curves $C_t\in \C_1 (q)$ incident with $(P,L)$ and tangent to $C$ at some point $P\in \A^2$. Moreover, for such a curve $C$ there are exactly $q-2$ curves $C_t$ tangent to $C$ at some $P\in \A^2$.   
\end{rem}

\begin{lem}[The structure of $8$--cycles]\label{Geo8}
Given a $4$--tuple of blocks of $\Bc_i(q)$ yielding a cycle of length $8$ in the incidence graph of $\I_i (q)$, the corresponding curves satisfy one of the following configuration.
\begin{enumerate}[(i)]
\item\label{8i} Two of the blocks are exceptional divisors $E_P$ and $E_Q$, with $P,Q\in \A^2 (\F_q)$ and the two other ones are curves $\widetilde{C}, \widetilde{D}\in \wC_i (q)$ such that the corresponding affine plane conics $C,D$ both contain $P$ and $Q$.
\item\label{8ii} One of the blocks is an exceptional divisor $E_P$ and the three other ones are curves $\widetilde{C}_1, \widetilde{C}_2$ and $\widetilde{C}_3$ such that $C_1, C_2 \ni P$ and $C_3$ is incident with a common flag $(P_{13},L_{13}) \in \Pc_i (q)$ with $C_1$ and to another common one $(P_{23}, L_{23})\in \Pc_i (q)$ with $C_2$.
\item\label{8iii} None of the blocks are exceptional divisors, they are curves $\widetilde{C}_1, \ldots , \widetilde{C}_4$. Moreover, there are $4$ flags $(P_{13}, L_{13}), (P_{14}, L_{14}), (P_{23}, L_{23}), (P_{24}, L_{24}) \in \Pc_i (q)$ such that the $P_{ij}$'s are distinct and $C_i, C_j$ are both incident with $(P_{ij}, L_{ij})$.
\end{enumerate}
  These three possible configurations are represented in figure \ref{Fig8cycl}.
\end{lem}

\begin{proof}
  Since two distinct exceptional divisors on $\B$ have no common point (see Remark \ref{disj}), a $4$--tuple of blocks yielding an $8$--cycle involves at most two exceptional divisors. The remaining curves are strict transforms of conics in $\C_i (q)$. Verifications are left to the reader that these conics should satisfy these conditions.
\end{proof}

\begin{figure}[h]
  \centering
\xymatrix{\includegraphics[scale=0.5]{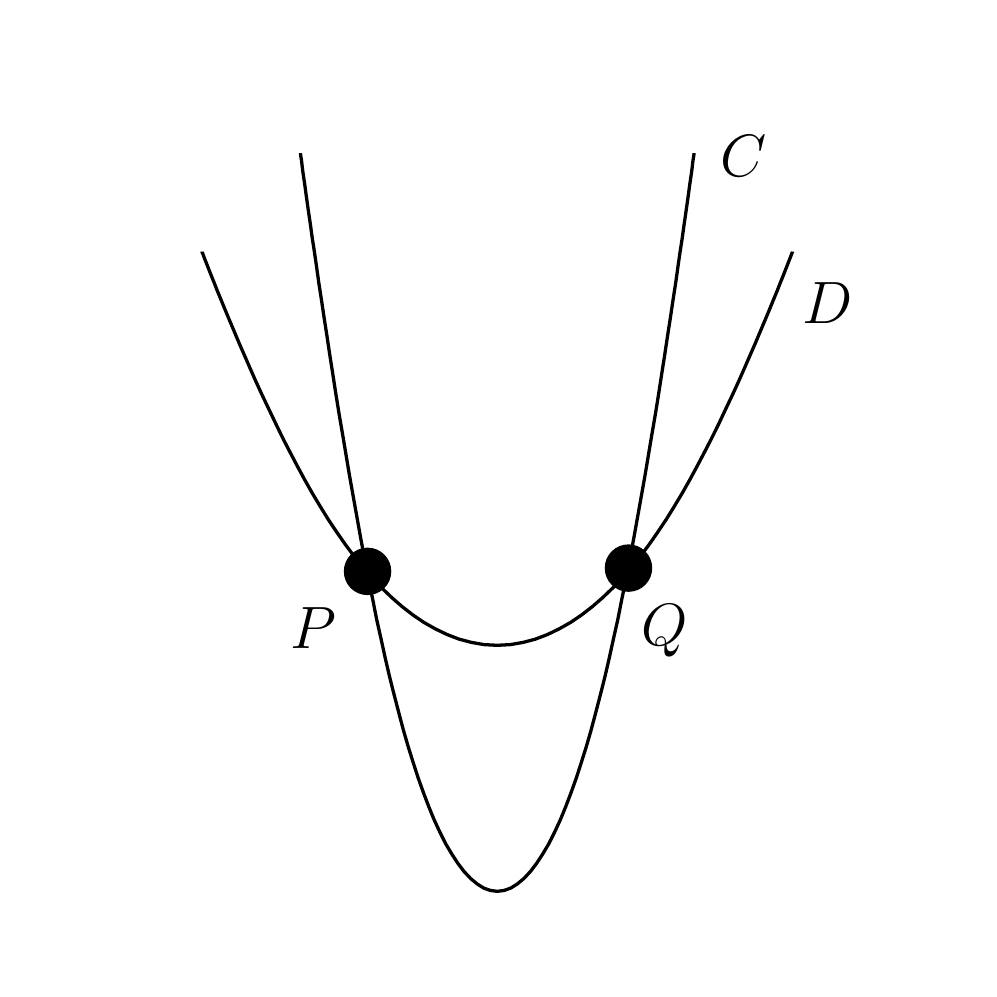} & \includegraphics[scale=0.5]{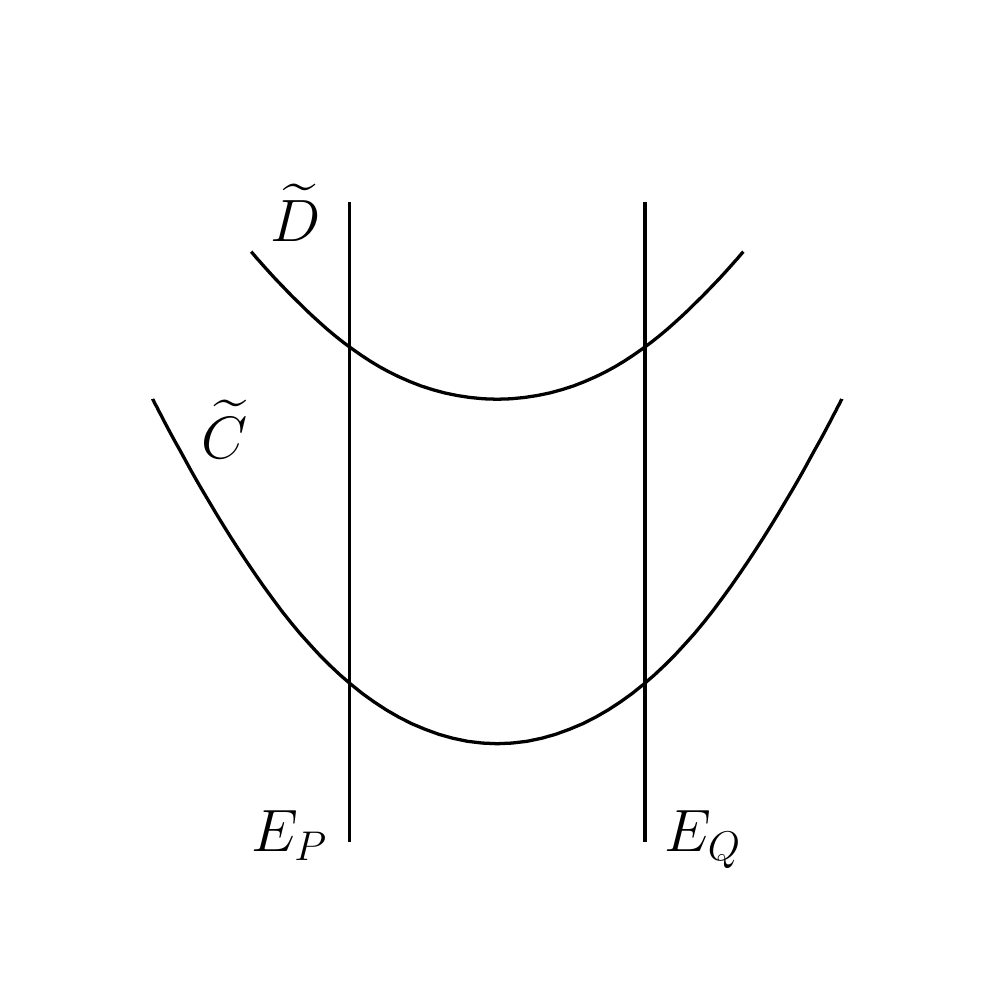} \ar[l]_{\pi} \\ 
\includegraphics[scale=0.5]{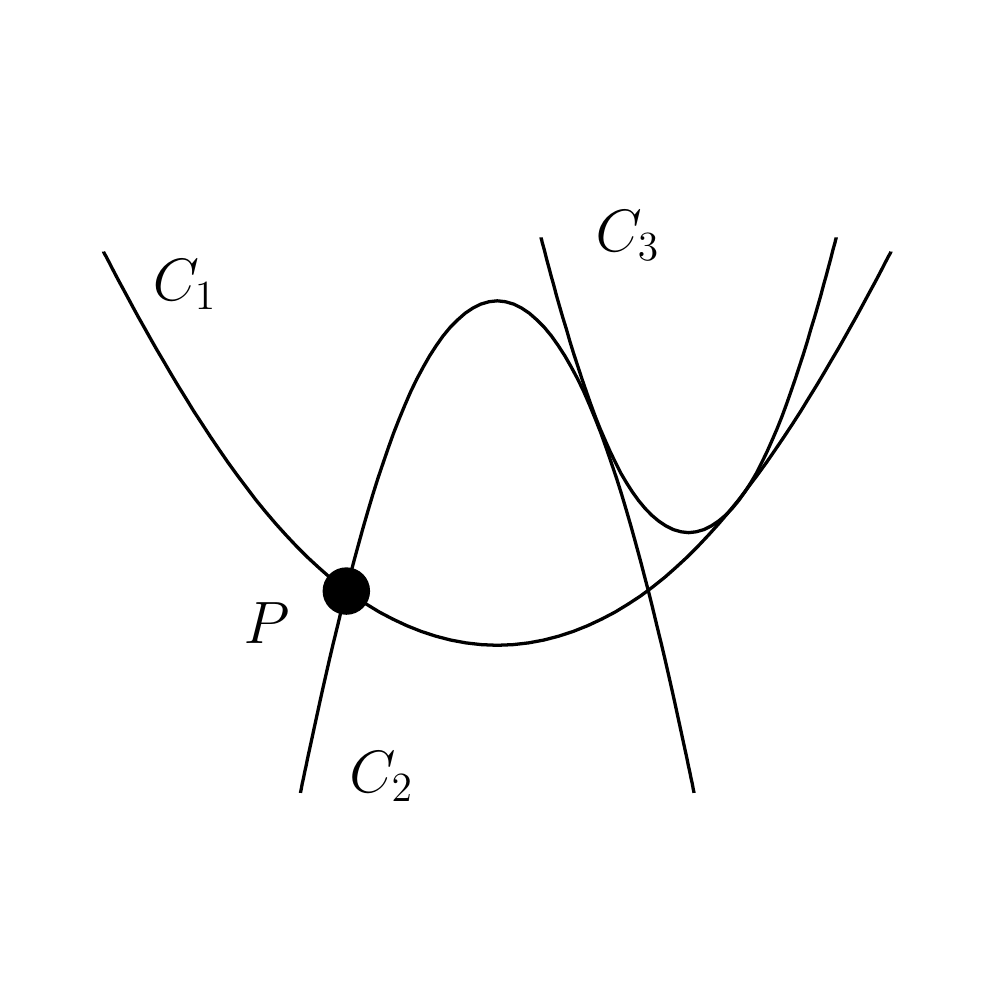} & \includegraphics[scale=0.5]{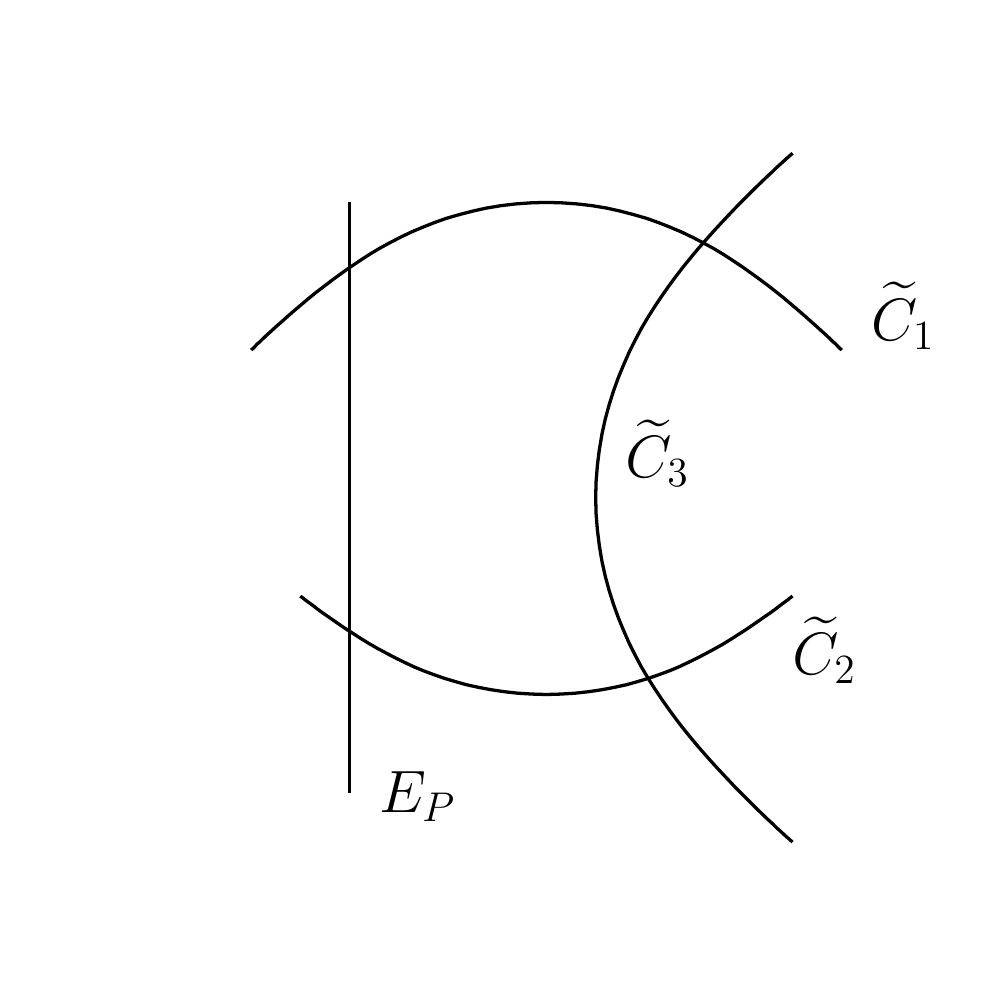} \ar[l]_{\pi} \\ 
\includegraphics[scale=0.5]{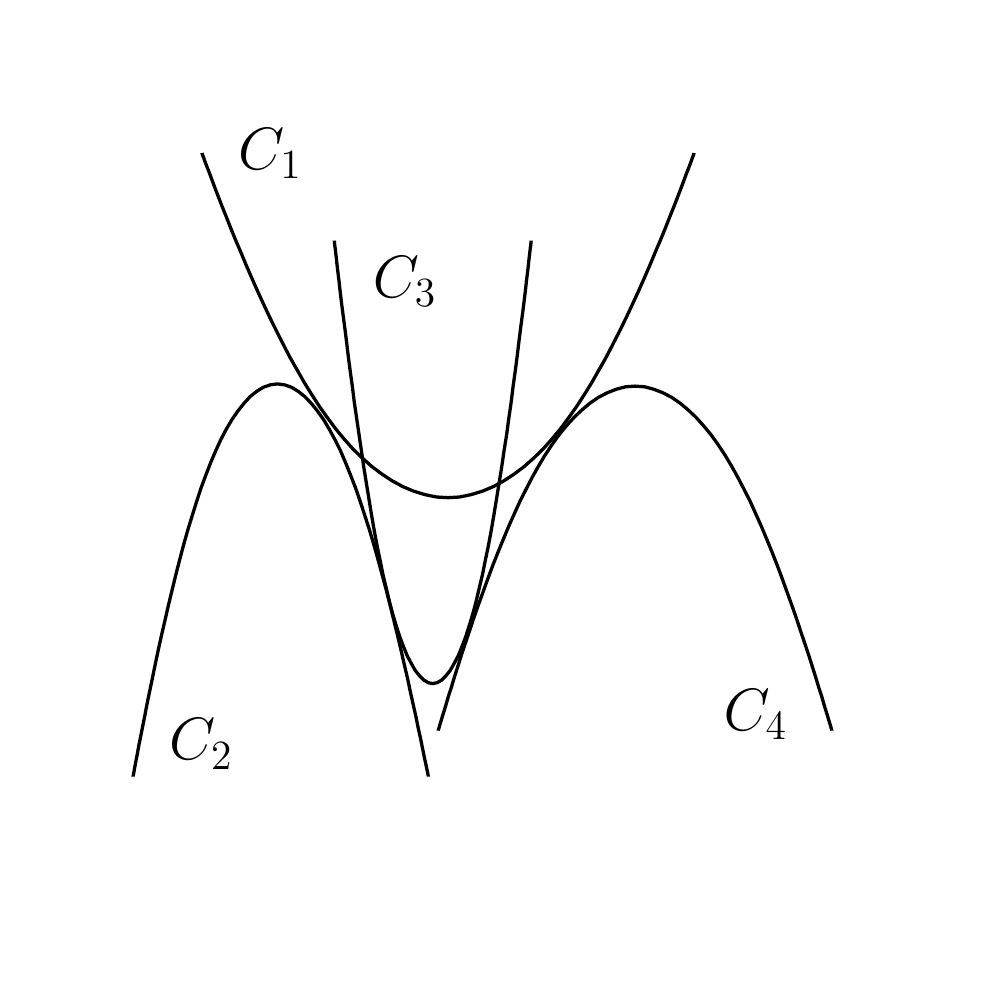} & \includegraphics[scale=0.5]{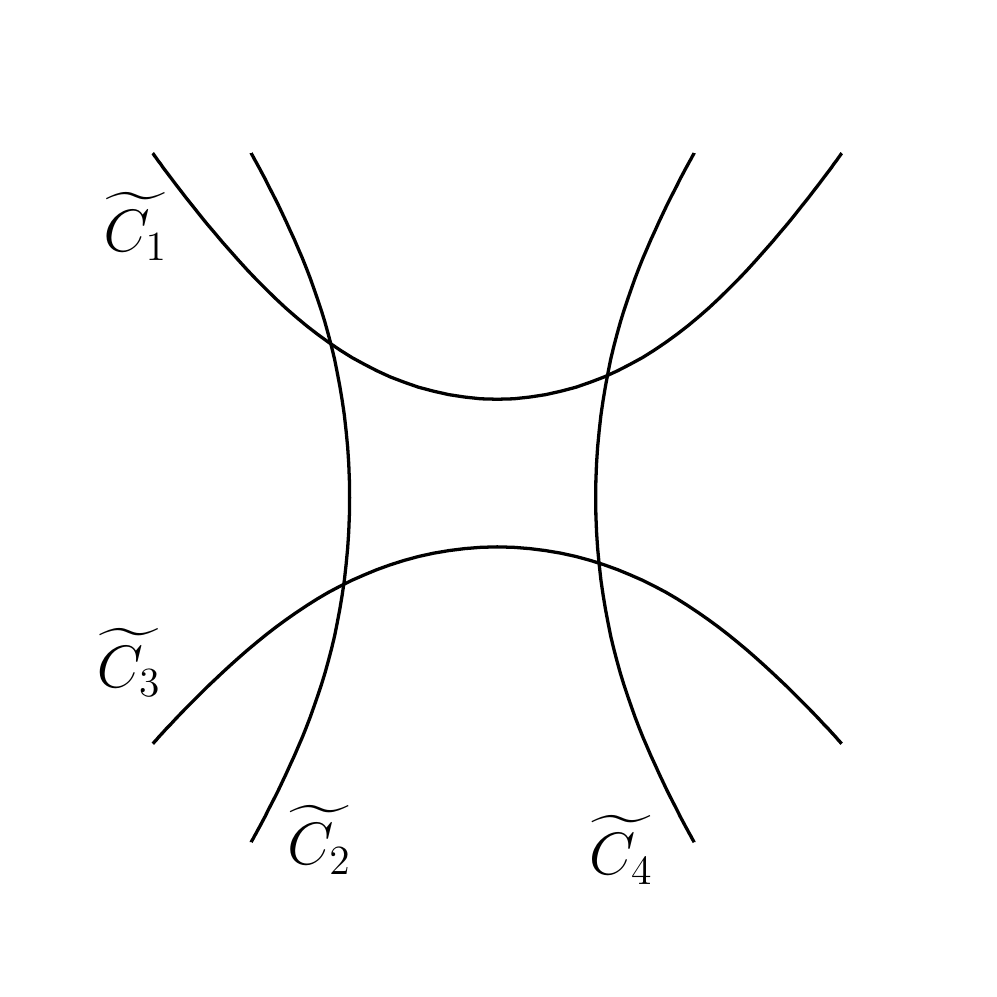} \ar[l]_{\pi}}      
  \caption{The three possible configurations of curves yielding $8$--cycles. Left hand pictures represent the curves in $\A^2$ and right hand ones the corresponding curves in $\B$ ($\pi$ denotes the Blow--up map).}
  \label{Fig8cycl}
\end{figure}

Now, using Lemmas \ref{6cycl}, \ref{8cycl}, \ref{Geo8} and Proposition \ref{MultiInc}, we can proceed to the proof of Theorem \ref{girth}.

\begin{proof}[Proof of Theorem \ref{girth}] 
 \emph{Step 1.} Theorem \ref{main} (\ref{m6}) entails the non-existence of cycles of length $4$ in the incidence graph. From Remark \ref{evengirth}, the girth is even and hence $\gamma (i,q)\geq 6$ for all pairs $(i,q)$.

\medbreak

\noindent \emph{Step 2.} 
Cycles of length $8$ always exist in the incidence graph: obviously, the case (\ref{8i}) of Lemma \ref{Geo8} always happens. Thus, $\gamma (i,q)\leq 8$ for all $i$ and all $q$.

\medbreak

\noindent \emph{Step 3.} Proposition \ref{MultiInc} together with Remark \ref{junction} entail that the incidence graph of $\I_i (q)$ has no $6$--cycles if $q$ odd and $i=1$ and if $q$ even and $i=2$ or $3$. Therefore, in these situations, $\gamma (i,q)=8$.

\medbreak

\noindent \emph{Step 4.} For the remaining situations, the girth of the incidence graph is actually $6$.
 To prove it, we explicit triples of elements of $\C_i$ which are in $(C3)$ configuration. In these three examples, the curves are denoted by $C_a, C_b, C_c$ and the flags by $(P_{ab}, L_{ab}), (P_{ac}, L_{ac})$ and $(P_{bc}, L_{bc})$.
Verifications are left to the reader.

\medbreak

\noindent \emph{Step 4.1. A triple of elements of $\C_1 (q)$ with $q$ even in $(C3)$ configuration.}
 Let $q$ be even.
By assumption, $q\geq 4$ (see \S \ref{context}).
Then, there exists $\eta \in \F_q$ such that $\eta \neq 0$ and $\eta \neq 1$. A $(C3)$ configuration is given by
$$
\begin{array}{cccl}
(C_a): & y & = & x^2 \\
(C_b): & y & = & \eta x^2 \\
(C_c): & y & = & (\eta +1)x^2+1  
\end{array}
; \quad
\begin{array}{ccl}
  P_{ab} & = & (0,0) \\
  P_{ac} & = & (\eta^{-1/2}, \eta^{-1})\\
  P_{bc} & = & (1,\eta)\\
\end{array}
; \quad
\begin{array}{cccl}
  L_{ab}: & y & = & 0 \\
  L_{ac}: & y & = & \eta^{-1}\\
  L_{bc}: & y & = & \eta.
\end{array}
$$

\medbreak

\noindent \emph{Step 4.2. A triple of elements of $\C_2 (q)$ with $q$ odd in $(C3)$ configuration.} 
$$
\begin{array}{cccl}
(C_a): & xy & = & x+y \\
(C_b): & xy & = & -x-y \\
(C_c): & xy & = & 4  
\end{array}
; \quad
\begin{array}{ccl}
  P_{ab} & = & (0,0) \\
  P_{ac} & = & (2,2)\\
  P_{bc} & = & (-2,-2)\\
\end{array}
; \quad
\begin{array}{cccl}
  L_{ab}: & y & = & -x \\
  L_{ac}: & y & = & -x+4\\
  L_{bc}: & y & = & -x-4.
\end{array}
$$

\medbreak

\noindent \emph{Step 4.3. A triple of elements of $\C_3 (q)$ with $q$ odd in $(C3)$ configuration.}
$$
\begin{array}{cccl}
(C_a): & x^2-\beta y^2 & = & x \\
(C_b): & x^2 - \beta y^2 & = & -x \\
(C_c): & x^2 - \beta y^2 & = & 1,  
\end{array}
; \quad
\begin{array}{ccl}
  P_{ab} & = & (0,0) \\
  P_{ac} & = & (1,0)\\
  P_{bc} & = & (-1,0)\\
\end{array}
; \quad
\begin{array}{cccl}
  L_{ab}: & y & = & 0 \\
  L_{ac}: & y & = & 1\\
  L_{bc}: & y & = & -1.
\end{array}
$$
\end{proof}

\subsection{Number of small cycles}\label{SmallCycles}
In the previous sub-section it is proved that the incidence graph of $\I_i(q)$ has either girth $6$ or girth $8$. In both cases we compute or bound above the number of cycles of minimum length.

\begin{thm}[Number of cycles of length $6$]\label{N6} For incidence structures $\I_i (q)$ whose incidence graph has girth $6$, let $N_6 (i,q)$ be the number of $6$--cycles. We have
$$
N_6 (i,q)\ \left\{
 \begin{array}{clcccc}
  = & q^3(q-1)^3(q-2)/6 &{\rm if} & i=1 & {\rm and} & q\ {\rm is\ even} \\
  \leq & q^2(q-1)(q^3-q^2-q)/3 &{\rm if} & i=2 & {\rm and} & q\ {\rm is\ odd} \\
  \leq & 2 q^4(q+1)(q-2) &{\rm if} & i=3 & {\rm and} & q\ {\rm is\ odd} \\
 \end{array}
\right. .
$$
In particular we always have $N_6(i,q)=\mathcal{O}(q^6)$.
\end{thm}

\begin{proof}
  From Lemma \ref{6cycl}, the number of $6$--cycles equals the number of non-ordered triples of conics in $(C3)$ configuration. For that, we compute the number of {\bf ordered} such triples and divide this number by $6$.

\medbreak

\noindent \emph{Step 1. For $\C_1(q)$, with $q$ even.}
We choose a flag $(P,L)$ and look for the number of ordered triples $C_a, C_b, C_c$ in $(C3)$ configuration and such that $C_a, C_b$ are both incident with $(P,L)$.

Remark \ref{CyclC1even} entails that there are $(q-1)^2$ choices for $C_c$ for which one can find $C_a, C_b$ both incident with $(P,L)$ and such that $C_a, C_b, C_c$ are in $(C3)$ configuration. We have:

\begin{itemize}
\item $(q-1)^2$ choices for $C_c$;
\item $(q-1)$ choices for $C_a$;
\item $(q-2)$ choices for $C_b$.
\end{itemize}
Since we also have $\sharp \I_1 (q)=q^3$ choices for $(P,L)$,
this yields $q^3(q-1)^3(q-2)$ ordered triples of conics in $(C3)$ configuration. Thus, there are $q^3(q-1)^3(q-2)/6$ non-ordered triples.

\medbreak

\noindent \emph{Step 2. For $\C_2 (q)$ with $q$ odd.}
As in the previous step, we choose a flag $(P,L)$ and look for triples $C_a, C_b, C_c$ in $(C3)$ configuration such that $C_a, C_b$ are both incident with $(P,L)$.
We choose an arbitrary curve $C_c\in \C_2 (q)$ avoiding $P$.
Without loss of generality, one can assume that $P=(0,0)$ and hence $C_c$ has an equation of the form $xy=ax+by+c$ with $c\neq ab$ and $c\neq 0$.
This yields $q^3-q^2-q$ possible choices for $C_c$.
Moreover, from Proposition \ref{MultiInc}, there are at most $2$ curves in $\C_2 (q)$ which
are incident with $(P,L)$ and have a common flag with $C_c$.
Thus, we have

\begin{itemize}
\item $\sharp \Pc_2 (q)\ =\ q^2(q-1)$ choices for $(P,L)$;
\item $q^3-q^2-q$ choices for $C_c$;
\item at most $2$ choices for $C_a$;
\item at most $1$ choice for $C_b$.
\end{itemize}
This yields at most $2q^2(q-1)(q^2-q^2-q)$ ordered triples and hence at most $q^2(q-1)(q^3-q^2-q)/3$ non-ordered triples.

\medbreak

\noindent \emph{Step 3. For $\C_3(q)$ with $q$ odd.} The approach is almost the same as that of the previous step.
Choose $(P,L)$ and $C_c$ avoiding $P$.
Without loss of generality, one can assume that $P=(0,0)$ and $C_c$ has an equation of the form
$x^2-\beta y^2=ax+by+c$ with $c\neq 0$ and $c\neq \frac{b^2}{4\beta}-\frac{a^2}{4}$. Since $\beta$ is a non-square in $\F_q$, the expression $\frac{b^2}{4\beta}-\frac{a^2}{4}$ is nonzero of all $(a,b)$. This entails that there are exactly $q^2(q-2)$ possible choices for $C_c$.

As previously, one applies Proposition \ref{MultiInc}. We have 
\begin{itemize}
\item $\sharp \Pc_3 (q)$ $=$ $q^2(q+1)$ choices for $(P,L)$;
\item $q^2(q-2)$ choices for $C_c$;
\item at most $4$ choices for $C_a$;
\item at most $3$ choices for $C_b$. 
\end{itemize}
This yields at most $12 q^4 (q+1)(q-2)$ ordered triples and hence at most $2q^4 (q+1)(q-2)$ non-ordered triples.
\end{proof}

To conclude the present section we state a result on the number of $8$--cycles.
Such a result can be obtained by counting the number of configurations describes in Lemma \ref{Geo8}. By counting one can obtain upper bounds on the number of these cycles. However the obtained formulas are pretty rough. We thus chose to give only the asymptotic behaviour of this number of $8$--cycles. Notice that the following theorem holds for even when the incidence graph has girth $6$.

\begin{thm}[Number of cycles of length $8$]\label{N8}
The number $N_8(i,q)$ of $8$--cycles of $\I_i (q)$ satisfies
$$
N_8(i,q) = \mathcal{O}(q^8).
$$
\end{thm}

\begin{proof}
  We have to count the number of configurations described in Lemma \ref{Geo8} (see also Figure \ref{Fig8cycl}).

\medbreak

\noindent \emph{Configurations (\ref{8i}).} There are $q^2$ rational points in $\A^2$. Thus, there are $\mathcal{O}(q^4)$ choices for $P,Q$. There are $\mathcal{O}(q)$ conics in $\C_i (q)$ containing $P$ and $Q$.
This yields $\mathcal{O}(q^2)$ choices for $C,D$.
Finally, we have $\mathcal{O}(q^6)$ configurations $(\ref{8i})$.

\medbreak

\noindent \emph{Configurations (\ref{8ii}).}
There are $\mathcal{O}(q^3)$ choices for $P$ and $\mathcal{O}(q^3)$ choices for $C_3$ (basically, ``almost all'' elements of $\C_i (q)$ avoid $P$).
Thus, we have $\mathcal{O}(q^6)$ choices for the pair $(P,C_3)$.
Choose two lines $L_1, L_2$ containing $P$ ($\mathcal{O}(q^2)$ choices), from Proposition \ref{MultiInc} there is at most $1$ element in $\C_i (q)$ incident with $(P, L_1)$ (resp. $(P, L_2)$) and sharing a common flag with $C_3$.
Thus, we have $\mathcal{O}(q^2)$ choices for $C_1, C_2$ and hence $\mathcal{O}(q^8)$ configurations (\ref{8ii}).

\medbreak

\noindent \emph{Configurations (\ref{8iii}).}
We have $\mathcal{O}(q^6)$ choices for $C_1, C_3$.
Choose two points $P,Q$ of $C_1$ ($\mathcal{O}(q^2)$ choices). From Proposition \ref{MultiInc}, there is at most one curve in $\C_i (q)$ which is incident with $(P, T_P C_1)$ (resp. $(Q, T_Q C_1)$). Thus there are $\mathcal{O}(q^8)$ configurations (\ref{8iii}).
\end{proof}

\section{LDPC codes from the incidence structures}\label{seccodes}

In this section, we construct and study LDPC codes from the previously defined incidence structures. Recall that, even if the incidence structures are constructed using geometry over an arbitrary finite field, the LDPC codes we construct are {\bf binary}.

Codes are defined in \S \ref{ddf}. The weights and minimum distance of these codes are studied in \S \ref{WandD}. The dimension of such codes is discussed in \S \ref{aboutDim}.

\subsection{The codes}\label{ddf}

\begin{defn}
Let $i\in \{1,2,3\}$.
The code $C(i,q)$ is the null space of the incidence matrix of $\I_i (q)$ having coefficients in $\F_2$.
\end{defn}

\begin{thm}\label{TTcodes}
The codes $C(1,q)$ (resp. $C(2,q)$, resp. $C(3,q)$) have a parity--check matrix of size $q^3 \times q^3$ (resp. $q^3 \times q^2(q-1)$ resp. $q^3 \times q^2(q+1)$).
  Moreover, these matrices are sparse and regular: each row has weight $q$ (resp. $q-1$, resp. $q+1$) and each column has weight $q$.
\end{thm}

\begin{proof}
 It is a straightforward consequence of Theorem \ref{main}.
\end{proof}

As a straightforward consequence of this Theorem, we obtain the length of these codes.

\begin{cor}[Length of $C(i,q)$]\label{length}
  The length $n(i,q)$ of the code $C(i,q)$ is
$$
n(i,q)= \left\{
  \begin{array}{ccc}
    q^3 & \textrm{if} & i=1;\\
    q^2(q-1) & \textrm{if} & i=2;\\
    q^2 (q+1) & \textrm{if} & i=3.
  \end{array}
\right.
$$
\end{cor}

\subsection{Weights and minimum distances}\label{WandD}

\begin{lem}
  For all pair $(i,q)$, the codewords of $C(i,q)$ have even weight.
\end{lem}

\begin{proof}
  Consider the incidence matrix $H(i,q)$ of $\I_i (q)$. A row of this matrix, is a parity check given by some block in $\Bc_i (q)$. Consider the $q^2$ rows corresponding to the exceptional divisors. Any two of these rows have disjoint supports (Remark \ref{disj}) and their sum is the vector $(1, \ldots, 1)$. It is an elementary exercise to prove that in a code whose parity--check matrix has a set of rows satisfying this property, the codewords always have even weight.  
\end{proof}

\begin{thm}[Minimum distance of $C(i,q)$]\label{mindist}
  The minimum distance $d$ of $C(i,q)$ is
  $$d = 2q.$$
\end{thm}

\medbreak

\noindent \textbf{Caution.}\label{Cau} In what follows, we deal with codewords of $C(i,q)$. Since we deal with binary codes (even if they arise from geometries over odd characteristic fields), a codeword of $C(i,q)$ can be regarded as a set of flags $\{(P_1, L_1), \ldots , (P_s,L_s)\}$ in $\Pc_i (q)$ such that the number of such flags incident with any block in $\Bc_i (q)$ is even.
From now on, we frequently use the representation of codewords as sets of flags in $\Pc_i (q)$.
Therefore, we allow ourselves to write $(P,L)\in c$ when $c$ is a codeword in $C(i,q)$ and its coordinate corresponding to $(P,L)$ equals $1$.

\begin{proof}[Proof of $d\geq 2q$]
  Let $c\in C (i,q)$ be a nonzero codeword. Let $(P,L)\in \Pc_i (q)$ be a flag such that $(P,L)\in c$ (regarding $c$ as a subset of $\Pc_i (q)$, see Caution above). Because of the block corresponding to the exceptional divisor $E_P$ above $P$, there is another line $L'$ such that $(P,L')\in \Pc_i (q)$ and $(P,L')\in c$.

Let $C_1, \ldots , C_{q-1}$ (resp. $C'_1, \ldots , C'_{q-1}$) be the conics in $\C_i (q)$ which are incident with $(P,L)$ (resp. $(P,L')$). An example is represented in Figure \ref{dmin}

\begin{figure}[!h]
  \centering
    \includegraphics[scale=0.5]{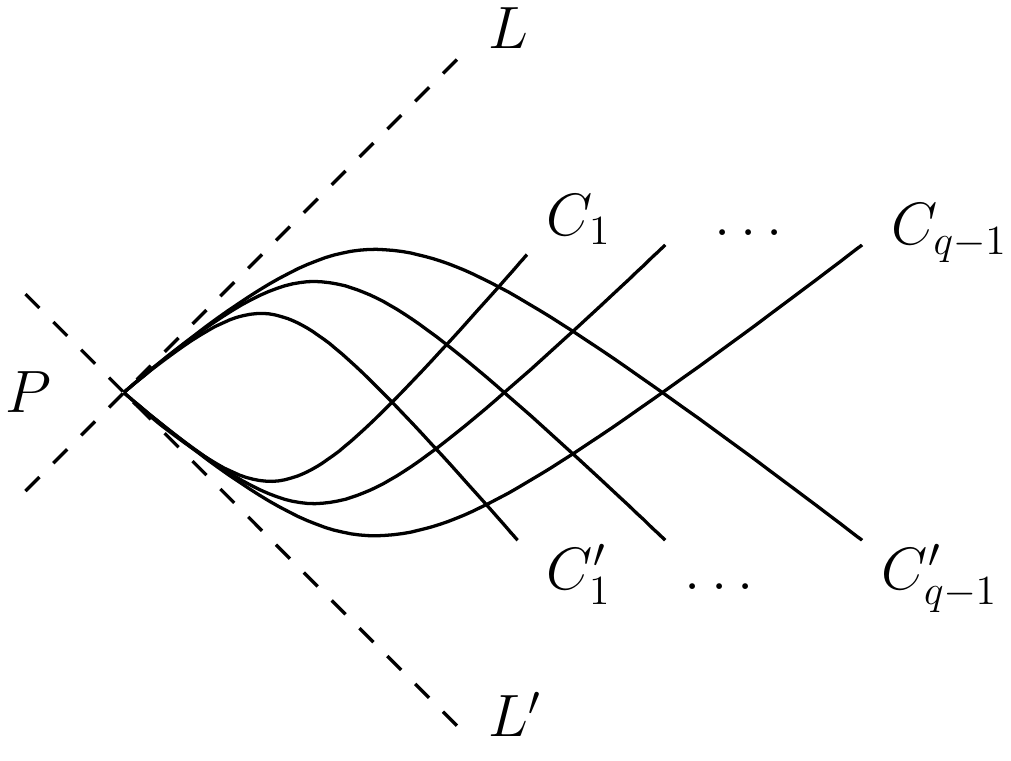}
  \caption{The curves $C_1, \ldots, C_{q-1}$ and $C'_1, \ldots , C'_{q-1}$}
  \label{dmin}
\end{figure}

Using Lemma \ref{incid}, one sees that any two of these $2q-2$ curves cannot share a common flag distinct from $(P,L)$ or $(P,L')$. Indeed, two curves of the form $C_i, C_j$ (resp. $C'_i, C'_j$) have a common flag $(P,L)$ (resp. $(P,L')$) and therefore do not meet at another point in $\A^2$. Two curves $C_i, C'_j$ meet with multiplicity $1$ at $P$ and hence cannot meet with multiplicity $>1$ at another point of $\A^2$.

Consequently, to satisfy all the parity checks, for all $i$ there is at least one flag $(P_i, L_i)$ (resp. $P'_i, L'_i$) incident with $C_i$ (resp. $C'_i$), distinct from $(P,L)$ and $(P', L')$ and contained in $c$. Moreover, the previous claim on the non incidence relations between the $C_i$'s and $C'_i$'s entails that the flags $(P_1, L_1), \ldots, (P_{q-1},L_{q-1}),(P'_1, L'_1),\ldots , (P'_{q-1}, L'_{q-1})$ are distinct to each other. This yields
$$
w(c)\geq \sharp \left\{(P,L), (P,L'), (P_1, L_1), \ldots, (P_{q-1},L_{q-1}),(P'_1, L'_1),\ldots , (P'_{q-1}, L'_{q-1}) \right\}=2q,
$$
where $w(c)$ denotes the Hamming weight of $c$.
\end{proof}

There remains to prove the existence of codewords of weight $2q$. Their existence and the construction of some of them is given in Appendix \ref{mini}.

\begin{rem}\label{WhyEx2}
  As noticed in Remark \ref{whyEx}, if the exceptional divisors were not in the block sets of the incidence structure, then the minimum distance would be only $\geq q$. It can be proved be reproducing the reasoning of the above proof.
Using the end of the proof in Appendix \ref{mini}, one can prove that the minimum distance would be actually $q$ in that case. 
\end{rem}

\subsection{Dimension}\label{aboutDim}
Unfortunately, we did not find formulas giving the dimension of the code $C(i,q)$ as a function of $q$. The dimension of the $C(i,q)$'s have been computed using {\sc Magma} for all prime power $q\leq 32$ (see tables 1 to 6 pages \pageref{T1} to \pageref{T6}). It turns out that the behaviour of these dimensions as a function of $q$ depends on the parity of $q$.
This claim is not surprising, we see for instance in Theorem \ref{girth}, that the girth of their Tanner graph already depends on the parity of $q$. Therefore the parity of $q$ has important consequences on the incidence structures.

\begin{rem}\label{deg3}
  It is worth noting that the length of $C(i,q)$ is of the form $q^3+\mathcal{O}(q^2)$. Therefore, if the dimension of the code for $q$ odd (resp. for $q$ even) is a polynomial in $q$, then this polynomial has degree at most $3$ and leading coefficient between $0$ and $1$.
\end{rem}

Using the tables in \S \ref{magma} and Lagrange's interpolation, we get the following conjectures.

\begin{conj}
  If $q$ is odd, then
$$
\dim C (1,q) = \frac{1}{2}q^3 -q^2+\frac{3}{2} q -1 \qquad
\dim C (2,q) = \frac{1}{2}q^3 - \frac{5}{2} q^2+\frac{9}{2} q -\frac{7}{2} \cdot
$$
The conjecture is satisfied for all odd prime power $5 \leq q \leq 31$.
\end{conj}

A more surprising fact on these dimensions is the following one.

\begin{lem}
  The dimension of $C(3,q)$ for $q$ odd is not a polynomial in $q$.
\end{lem}

\begin{proof}
  From remark \ref{deg3}, if the dimension is a polynomial, this polynomial has degree at most $3$. Using Table 5 page \pageref{T5}  and interpolating the values for $q=5,7,9,11$, we obtain the polynomial $\frac{23}{48}q^3-\frac{15}{16}q^2-\frac{215}{48} q+\frac{239}{16}$. It is easy to check that this polynomial does not give the other computed values of the dimension.
\end{proof}

Looking at the calculations presented in \S \ref{magma}, the codes $C(i,q)$ seem have asymptotic information rate $1/2$ when $q$ is odd. This rate seems to be higher when $q$ is even. Since the constructed parity--check matrices are almost square, they are redundant. Mostly, they have about twice more rows than necessary.

\subsection{Cycles of the Tanner graph}
An important criterion for the efficiency of LDPC codes is the girth of their Tanner graph and the number of small cycles. It is proved in Theorem \ref{girth} that the girth of their Tanner graphs are $6$ or $8$. Moreover, Theorem \ref{N6} asserts that if the girth is $6$, then the number of small cycles is $\mathcal{O}(q^6)$ and hence $\mathcal{O}(n^2)$, where $n$ denotes the length of the code.
Theorem \ref{N8} asserts that the number of cycles of length $8$ is $\mathcal{O}(q^8)$ and hence $\mathcal{O}(n^{8/3})$.


\subsection{Calculations}\label{magma}
We developed \textsc{Magma} programs producing the codes $C(i,q)$. These programs are available on \url{http://www.lix.polytechnique.fr/Labo/Alain.Couvreur/doc_rech/LDPC_codes.tar.gz}.
Thanks to them we are able to calculate by computer the dimensions of the $C(i,q)$'s for $q\leq 32$. As seen in \S \ref{secgirth}, the parity of $q$ has an important influence on the incidence structure, and hence on the codes. Therefore, we present in distinct tables the cases $q$ even and $q$ odd.

\bigbreak

\begin{center}
\begin{tabular}{|c|c|c|c|c|c|c|}
  \hline
 $q$  & Length & Number of  & Minimum & Girth & Dimension & Rate \\
  & & Parity checks  & distance &  &  & (approx.) \\
\hline
5 & 125 & 125 & 10 & 8 & 44 & 0,35 \\
\hline
7 & 343 & 343 & 14 & 8 & 132 & 0,38 \\
\hline
9 & 729 & 729 & 18 & 8 & 296 & 0,41 \\
\hline
11 & 1331 & 1331 & 22 & 8 &560 & 0,42 \\
\hline
13 & 2197 & 2197 & 26 & 8 & 948 & 0,43 \\
\hline 
25 & 15625 & 15625 & 50 & 8 & 7224 & 0,46 \\
\hline
31 & 29791 & 29791 & 62 & 8 & 13980 & 0,47 \\
\hline
\end{tabular}  
\bigbreak
\textsc{Table 1.} The codes $C(1,q)$ for $q$ odd.
\label{T1}
\end{center}

\bigbreak

\begin{center}
    \begin{tabular}{|c|c|c|c|c|c|c|}
        \hline
 $q$  & Length & Number of  & Minimum & Girth & Dimension & Rate \\
  & & Parity checks  & distance &  &  & (approx.) \\
\hline
 4 & 64 & 64 & 8 & 6 & 23 & 0,36 \\
\hline
8 & 512 & 512 & 16 & 6 & 259 & 0,51 \\
\hline
16 & 4096 & 4096 & 32 & 6 & 2615 & 0,63 \\
\hline
32 & 32768 & 32768 & 64 & 6 & 24151 & 0,74 \\
\hline
    \end{tabular}

\bigbreak

 \textsc{Table 2.} The codes $C(1,q)$ for $q$ even.
 \end{center}

\bigbreak
\begin{center}
    \begin{tabular}{|c|c|c|c|c|c|c|}
        \hline
 $q$  & Length & Number of  & Minimum & Girth & Dimension & Rate \\
  & & Parity checks  & distance &  &  & (approx.) \\
\hline
5 & 100 & 125 & 10 & 6 & 19 & 0,19 \\
\hline
7 & 294 & 343 & 14 & 6 & 77 & 0,22 \\
\hline
9 & 648 & 729 & 18 & 6 & 199 & 0,31 \\
\hline
11 & 1210 & 1331 & 22 & 6 & 409 & 0,34 \\
\hline
13 & 2028 & 2197 & 26 & 6 & 731 & 0,36 \\
\hline
25 & 15000 & 15625 & 50 & 6 & 6359 & 0,42 \\
\hline 
31 & 28830 & 29791 & 62 & 6 & 12629 & 0,44 \\
\hline 
    \end{tabular}

\bigbreak
\textsc{Table 3.} The codes $C(2,q)$ for $q$ odd.
  \end{center}

\bigbreak

  \begin{center}
    \begin{tabular}{|c|c|c|c|c|c|c|}
        \hline
 $q$  & Length & Number of  & Minimum & Girth & Dimension & Rate \\
  & & Parity checks  & distance &  &  & (approx.) \\
\hline
 4 & 48 & 64 & 8 & 8 & 11 & 0,23 \\
\hline
8 & 448 & 512 & 16 & 8 & 176 & 0,39 \\
\hline
16 & 3840 & 4096 & 32 & 8 & 2001 & 0,52 \\
\hline
32 & 31744 & 32768 & 64 & 8 & 19594 & 0,62 \\
\hline
    \end{tabular}

\bigbreak

 \textsc{Table 4.} The codes $C(2,q)$ for $q$ even.
  \end{center}

\bigbreak

  \begin{center}
    \begin{tabular}{|c|c|c|c|c|c|c|}
        \hline
 $q$  & Length & Number of  & Minimum & Girth & Dimension & Rate \\
  & & Parity checks  & distance &  &  & (approx.) \\
\hline
5 & 150 & 125 & 10 & 6 & 29 & 0,19 \\
\hline
7 & 392 & 343 & 14 & 6 & 102 & 0,26 \\
\hline
9 & 810 & 729 & 18 & 6 & 248 & 0,31 \\
\hline
11 & 1452 & 1331 & 22 & 6 & 490 & 0,34 \\
\hline
13 & 2366 & 2197 & 26 & 6 & 852 & 0,36 \\
\hline
17 & 5202 & 4913 & 34 & 6 & 2032 & 0,39 \\
\hline
25 & 16250 & 15625 & 50 & 6 & 7513 & 0,46 \\
\hline 
31 & 30752 & 29791 & 62 & 6 & 14431 & 0,47 \\
\hline 
    \end{tabular}

\bigbreak
\label{T5}
\textsc{Table 5.} The codes $C(3,q)$ for $q$ odd.
  \end{center}


\begin{figure}[b]
  \centering
\includegraphics[scale=1.1]{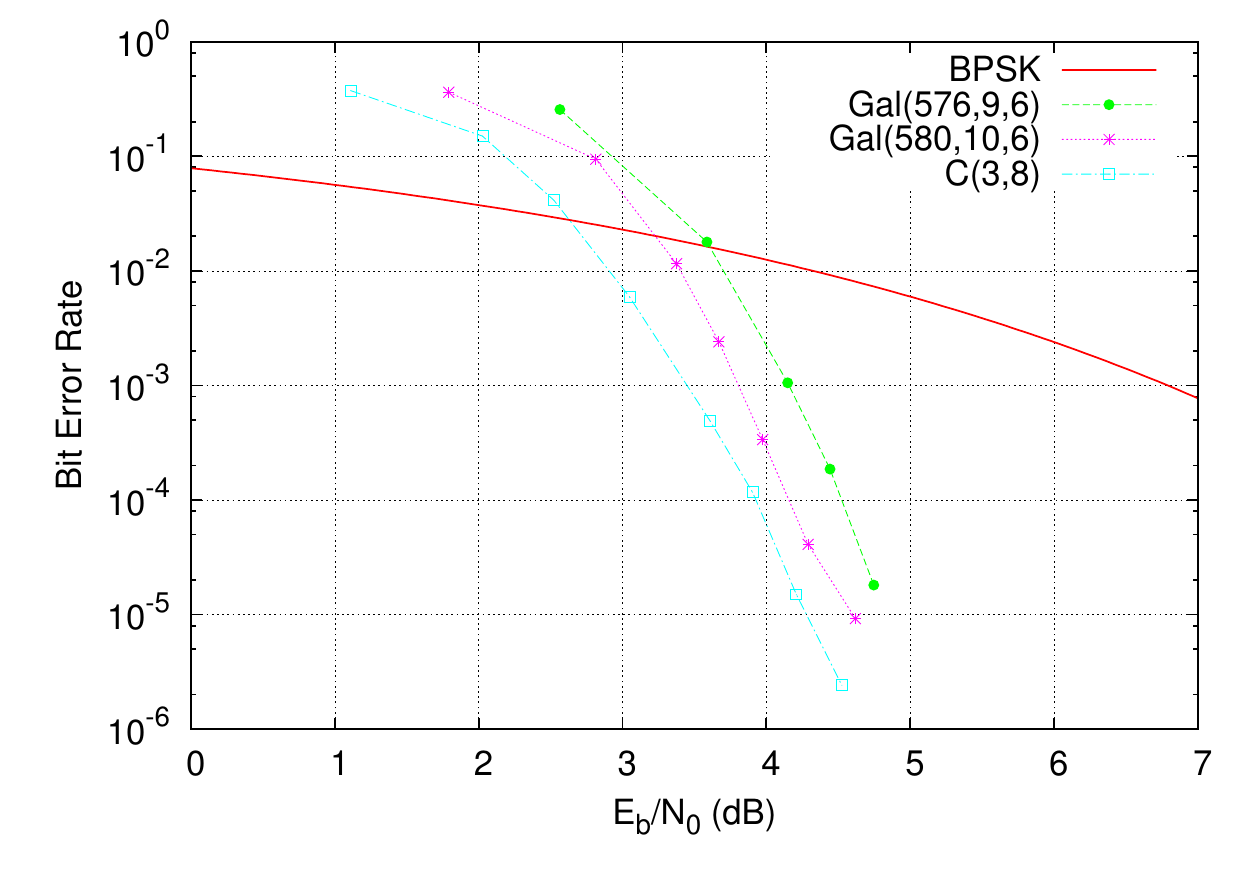}
  \caption{Decoding performances of the code $C(3,8)$ (with parameters $[576,233]$) and two Gallager codes, with respective row weights $9$ and $10$ and parameters $[576,197]$ and $[580,237]$.}
\label{FigC_3_8}  
\end{figure}

  \begin{center}
    \begin{tabular}{|c|c|c|c|c|c|c|}
        \hline
 $q$  & Length & Number of  & Minimum & Girth & Dimension & Rate \\
  & & Parity checks  & distance &  &  & (approx.) \\
\hline
 4 & 80 & 64 & 8 & 8 & 19 & 0,24 \\
\hline
8 & 576 & 512 & 16 & 8 & 223 & 0,39 \\
\hline
16 & 4352 & 4096 & 32 & 8 & 2223 & 0,51 \\
\hline
32 & 33792 & 32768 & 64 & 8 & 21575 & 0,64 \\
\hline
    \end{tabular}

\bigbreak

    \textsc{Table 6.} The codes $C(3,q)$ for $q$ even.
   \label{T6}
  \end{center}

\subsection{Simulations on the Gaussian Channel}
Using the function \texttt{LDPCSimulate} of \textsc{Magma}, we have simulated the performances of some codes $C(i,q)$ on the Additive White Gaussian Noise Channel.
We compare these results with the performances of regular Gallager codes having nearly the same rate and row weight.
These results are presented in Figures \ref{FigC_3_8}, \ref{FigC_1_13} and \ref{FigC_2_16}. The performances of our codes turn out to beat those of Gallager codes having similar length, rate and row weight.

\begin{figure}[b]
  \centering
\includegraphics[scale=1.1]{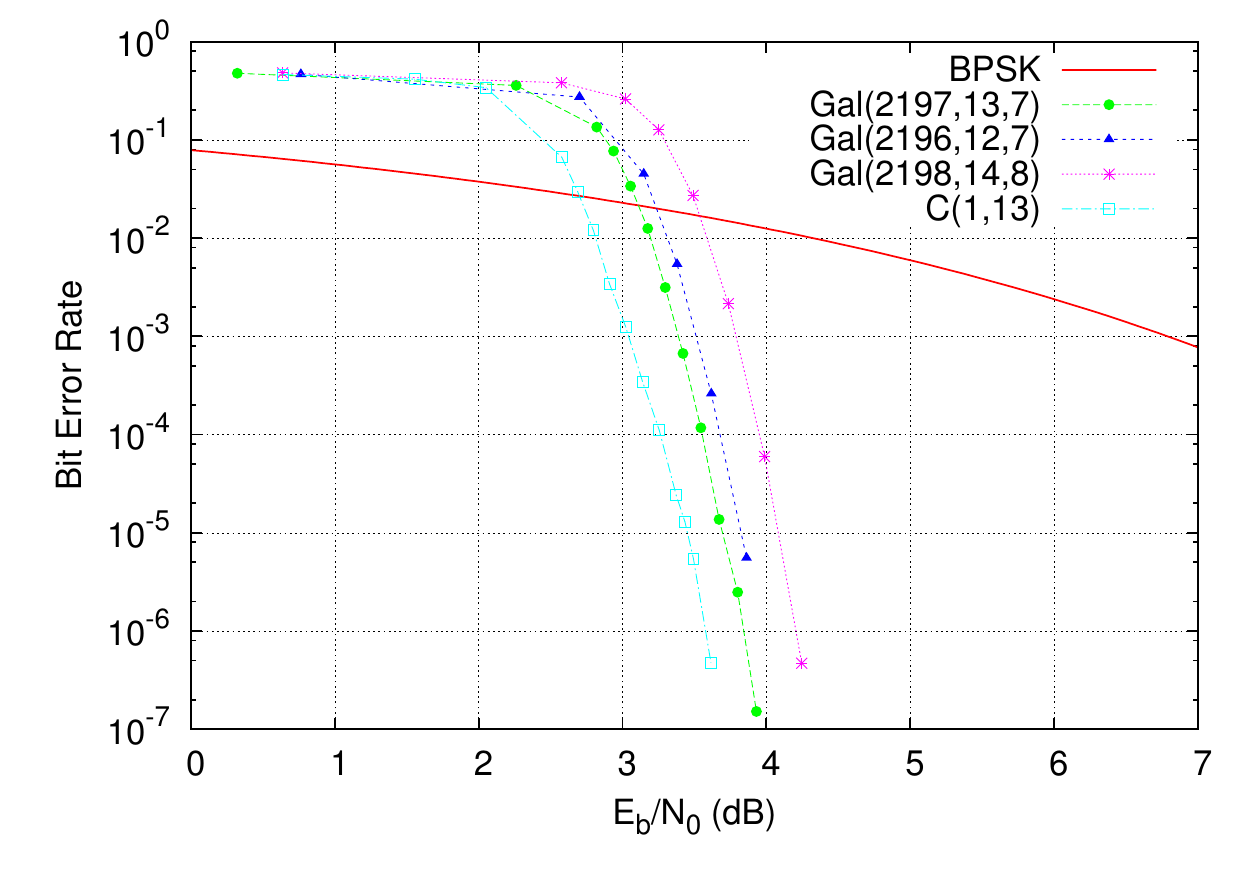}
  \caption{Decoding performances of the code $C(1,13)$ (with parameters $[2197, 948]$) and three Gallager codes, with respective row weights $12$, $13$ and $14$ and parameters $[2196,921]$, $[2197, 1020]$ and $[2198,949]$.}
\label{FigC_1_13}  
\end{figure}


\subsubsection{Details of the simulations}
All the bit error rates above $10^{-4}$ have been obtained after between $10^4$ and $10^5$ random tests. For Bit error rates under $10^{-4}$, between $10^6$ and $10^7$ random tests are done.
The number of iterations of the iterative decoding algorithm is set to $50$ for the simulations presented in Figure \ref{FigC_3_8} and to $500$ for the simulations in Figures \ref{FigC_1_13} and \ref{FigC_2_16}.

\begin{figure}[!]
  \centering
\includegraphics[scale=1.1]{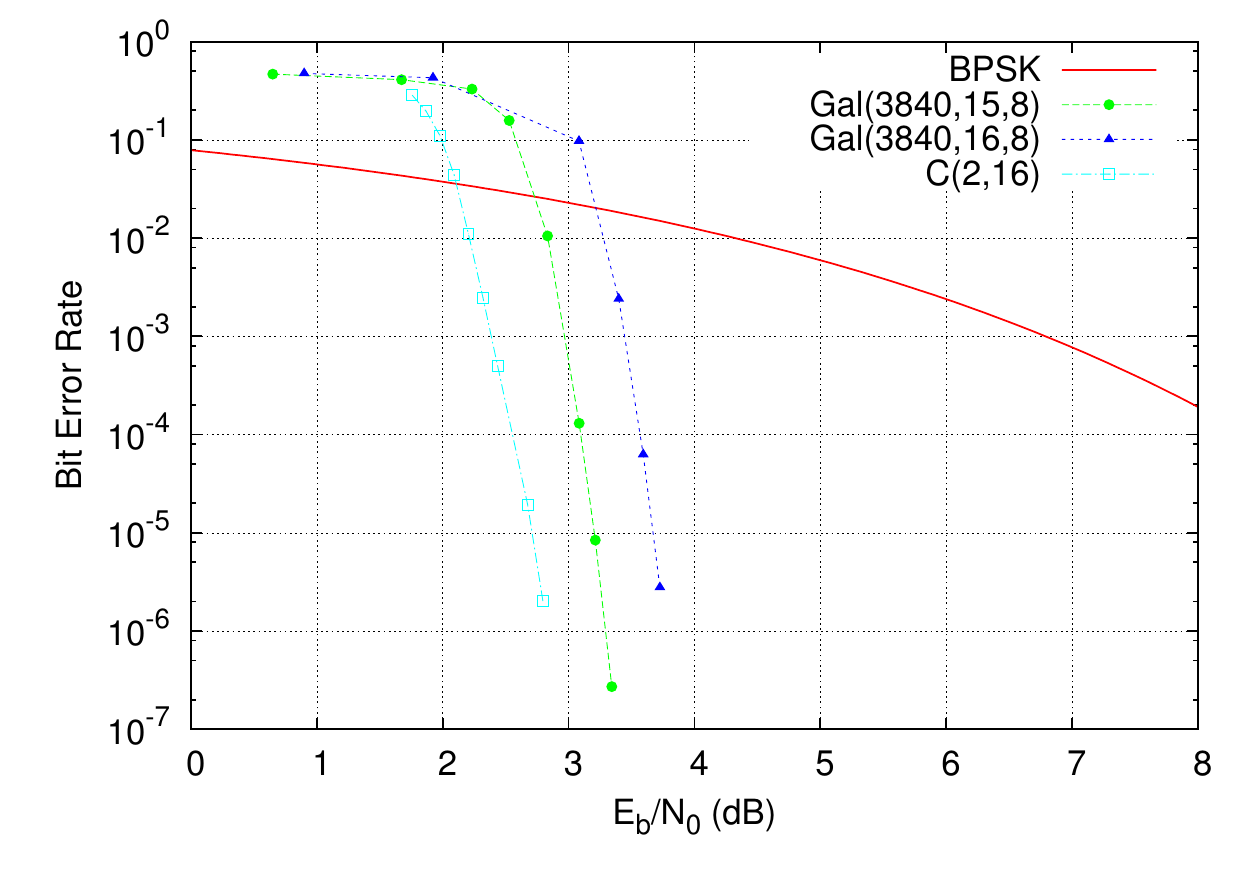}
  \caption{Decoding performances of the code $C(2,16)$ (with parameters $[3840, 2001]$) and two Gallager codes, with respective row weights $15$ and $16$ and parameters $[3840,2041]$ and $[3840,1927]$.}
\label{FigC_2_16}  
\end{figure}

\appendix

\section{Automorphisms of the plane}\label{AppA}

\begin{proof}[Proof of Lemma \ref{Aut} (\ref{Aut1})] Let $(P, \overline{P}, P_3, P_4)$ and $(Q, \overline{Q}, Q_3, Q_4)$ be two such $4$--tuples. Obviously, there exists a unique $\sigma \in \PGL (3, \F_{q^2})$ sending $(P, \overline{P}, P_3, P_4)$ onto $(Q, \overline{Q}, Q_3,$ $Q_4)$. Let us prove that $\sigma$ is actually defined over $\F_q$, i.e. that $\bar{\sigma}=\sigma$, where $\bar{\sigma}$ denotes the conjugate of $\sigma$ under the Frobenius action.
Since $Q_3$ is rational, we have $Q_3=\overline{Q}_3$ and hence $Q_3=\sigma P_3=\overline{\sigma P_3}=\bar{\sigma} P_3$. In the same way, we obtain $\bar{\sigma} P_4=Q_4$. Moreover $\overline{\sigma P}=\overline{Q}=\bar{\sigma} \overline{P}$. By the same manner, we prove that $\bar{\sigma}P=Q$. By uniqueness of $\sigma$, we get $\sigma=\bar{\sigma}$ and hence $\sigma \in \PGL (3, \F_q)$.
\end{proof}

\begin{proof}[Proof of Lemma \ref{Aut} (\ref{Aut2})] let $(P_1, P_2 ,L_1, L_2)$ and $(Q_1, Q_2, M_1, M_2)$ be two such $4$--tuples. Choose two rational points $P_1', P_2'$ such that $P_1'\in L_1$,  $P_2' \in L_2$, $P_1' \notin L_2$ a $P_2' \notin L_1$. Choose two rational points $Q_1', Q_2'$ satisfying the same conditions with respect to the triple $(Q_1, Q_2, M_1, M_2)$. There exists a unique $\sigma \in \PGL (3, \F_q)$ sending $(P_1,P_2, P_1', P_2')$ onto $(Q_1, Q_2, Q_1', Q_2')$.
\end{proof}
  
The proofs of Lemma \ref{Aut} (\ref{Aut3}) and (\ref{Aut4}) are obtained using the same techniques as in the above proofs.

\section{Minimum weight codewords}\label{mini}

In \S \ref{WandD}, it is proved that the codes $C(i,q)$ have minimum distance at least $2q$. In this appendix, we give an explicit construction of some codewords of weight $2q$ which concludes the proof of Theorem \ref{mindist}.
To construct such codewords, we have to introduce additional mathematical tools.

\begin{lem}\label{TH}
  For all $u\in \F_q^2$, all $P\in \A^2(\F_q)$ and all $a\in \F_q\setminus \{0\}$, the sets $\C_i (q)$ are globally preserved by the translation of vector $u$ and by the homotecy centred at $P$ with ratio $a$. These affine automorphisms induce therefore automorphisms of $\I_i (q)$.
\end{lem}

\begin{proof}
  It is sufficient to prove that these automorphisms of $\A^2$ prolongated and regarded as automorphisms of $\P^2$ leave invariant any point at infinity.
Since translations and homothecies send a line onto a parallel one, they fix any point at infinity. 
\end{proof}

\begin{nota}[Parallel lines]\label{para}
  If two lines $L,L'\in \A^2$ are parallel (i.e. they do not meet in $\A^2$) we write $L\sim L'$.
Moreover the class of a line $L$ modulo $\sim$ is denoted by $[L]$. The set of such classes is isomorphic to $\P^1$. The class of vertical lines is denoted by $[V]$ and that of horizontal lines by $[H]$.
\end{nota}

\begin{prop}\label{Corr}
Let $i\in \{1,2,3\}$.
Let $L_0, L$ be two lines in $\A^2$ meeting at $P\in \A^2 (\F_q)$. If $i=1$, then $L,L_0$ are assumed to be non vertical; if $i=2$, then they are assumed to be neither vertical nor horizontal.
Let $C\in \C_i (q)$ be a curve incident with $(P,L)$ and $Q$ be the other point of intersection of $C$ with $L_0$. Finally, denote by $M$ the line $M:=T_Q C$.  
Then,
 \begin{enumerate}[(i)]
  \item the class $[M]$ modulo $\sim$ (see Notation \ref{para}) depends only on $[L]$ and $[L_0]$ and neither on $P$, nor on $C$.
  \item the map $[L] \mapsto [M]$ is an involution $\psi_{[L_0]}$ of $\P^1\setminus \{[V],[L_0]\}$ if $i=1$, of $\P^1\setminus \{[V],[H],[L_0]\}$ if $i=2$ and of $\P^1\setminus \{[L_0]\}$ if $i=3$.
  \end{enumerate}
See figure \ref{invo} for an illustration.
\end{prop}

\begin{figure}[!h]
  \centering
    \includegraphics[scale=0.5]{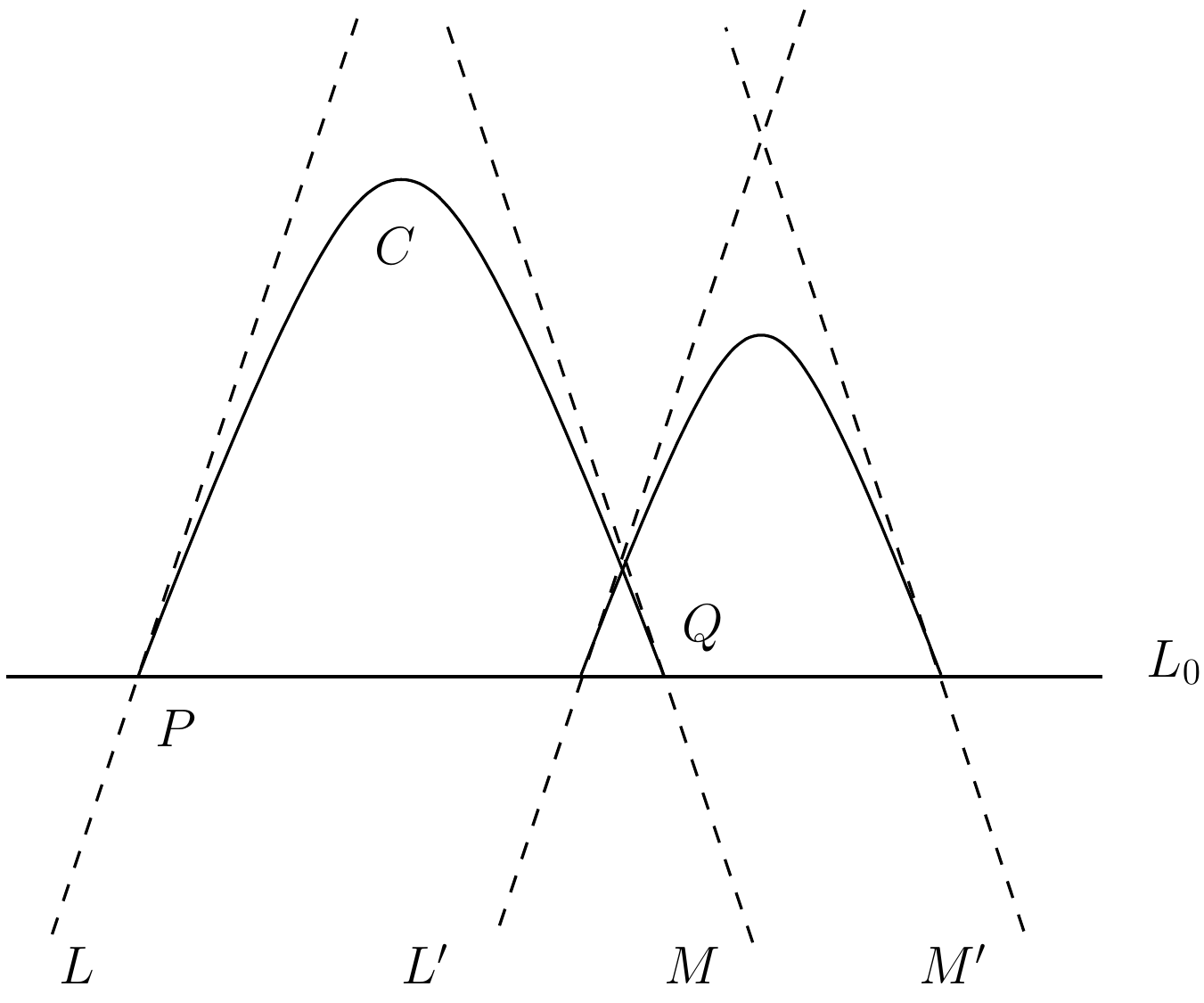}
  \caption{The involution $\psi_{[L_0]}$. In this picture, $L \sim L'$ and $\psi_{[L_0]}([L])=[M]=[M']$.}
\label{invo}
\end{figure}

\begin{proof}
First, notice that $C$ is the unique element of $\C_i (q)$ incident with $(P,L)$ and containing $Q$. Indeed, the existence of two such distinct curves $C,D$ would yield a contradiction with Lemma \ref{incid}, since $C,D$ would meet at least once at $Q$ and twice at $P$.

Let $C'\in \C_i (q)$ be another curve incident with $(P,L)$ and $Q'$ be the other point of intersection of $C'$ with $L_0$. 
By the same way $C'$ is the unique element of $\C_i (q)$ incident with $(P,L)$ and containing $Q'$.
Let $h$ be the homotecy of centre $P$ sending $Q$ on $Q'$. By uniqueness, $h$ sends $C$ onto $C'$ and $T_{Q'}C'=h(T_Q C)=M$, thus $T_{Q}C' \sim M$. Therefore $[M]$ does not depend on the choice of $C$.

Afterwards, let $P'\in L_0$ be another point and $L'\sim L$ be a line containing $P'$. Let $t$ be the translation sending $P$ on $P'$, this map sends $C$ onto a curve incident with $(P',L')$. The curve $t(C)$ meets $L_0$ at $t(Q)$ and its tangent at this point is parallel to $M$. This shows that $[M]$ does not depend on $P$.

Finally, to prove that this correspondence is an involution, it is sufficient to show that $[M]$ is sent onto $[L]$, which is obvious since $C$ is incident with $(Q,M)$, meets $L_0$ at another point $P$ and has $L$ a tangent at this point, thus $\psi_{[L_0]}([L])=[M]$.
\end{proof}

\noindent \textbf{Construction of minimum weight codewords.} Using Proposition \ref{Corr}, one can prove the existence of codewords of weight $2q$ in $C(i,q)$ and construct them explicitly.
Choose a line $L_0$ whose class in $\P^1$ is distinct from $[V]$ if $i=1$ and distinct from $[V], [H]$ if $i=2$. The involution $\psi_{[L_0]}$ introduced in Proposition \ref{Corr} is either constant\footnote{One can prove that the involution is constant when $i=1$ and $q$ is even.  In even characteristic, the tangents of a curve of equation $y=ax^2+bx+c$ are all parallel!} or permutes at least two distinct classes $[L], [M]$. If it is constant, then choose an arbitrary pair of classes $[L],[M]$, else choose $[L]\neq [M]$ such that $\psi_{[L_0]}([L])=[M]$.
Recall that words in $\F_2^n$ can be represented by sets of flags (see Caution page \pageref{Cau}).
Consider the word in $\F_2^n$ defined by
$$
\begin{array}{ccl}
c & := & \{(P,L_P) : P\in L_0,\ L_P\ni P\ {\rm and}\ L_P\sim L\} \cup\\
 & & \null \qquad \qquad \{(P, M_P): P\in L_0,\ M_P\ni P\ {\rm and}\ M_P\sim M\}.
\end{array}
$$
See figure \ref{c} for an illustration.

\begin{figure}[!]
  \centering
  \includegraphics[scale=0.5]{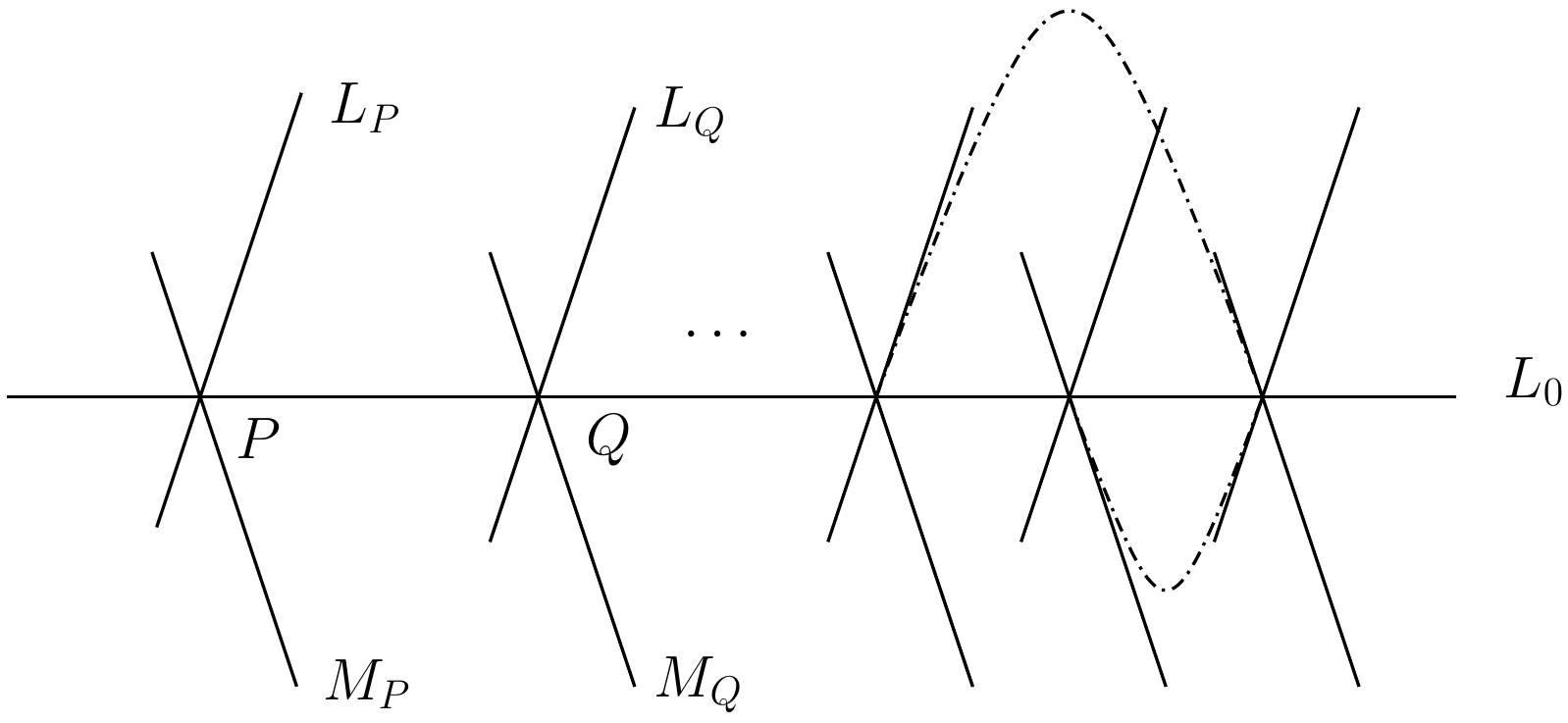}
  \caption{The codeword $c$.}
\label{c}
\end{figure}

The line $L_0$ has $q$ rational points in $\A^2$ and the above word is given by $2$ flags per point in $L_0$. Thus, it has weight $2q$.
There remain to show that it is a codeword of $C(i,q)$, which means that any block of $\Bc_i(q)$ is always incident with an even number of flags in $c$.
\begin{enumerate}[(1)]
\item Let $E_P$ be an exceptional divisor. If $P\notin L_0$ then none flag in $C$ is incident with $E_P$. Else, exactly two of them are, namely $(P,L_P)$ and $(P, M_P)$.
\item Let $C \in \C_i (q)$, if $C$ is incident with a flag $(P, L_P)$ in $c$, then $C$ meets $L_0$ at another point $Q$.
If the involution $\psi_{[L_0]}$ is constant, then $T_Q C\sim L_P$ and hence $T_Q C = L_Q$, else $T_Q C =M_Q$. In both cases, if $C$ is incident with an element of $c$, then it is always incident with a second one.
\end{enumerate}
This concludes the proof.

\section{Index of notations and terminologies}\label{AppC}

\begin{center}
\begin{tabular}{lp{2cm}r}
$m_P(C,D)$ & & \S \ref{IntMult}
\\
The line $(PQ)$ &  & Nota \ref{notdtes} 
\\
$\Lambda_1 (P_1, P_2, P_3, L)$  & &      Lem \ref{contrick}
 \\
$\Lambda_2 (P_1, P_2, L_1, L_2)$  & &      Lem \ref{contrick2}
 \\
$L_{\infty}, \alpha, P_{\infty}, Q_{\infty}, R_{\infty}, \overline{R}_{\infty}$  & &     Nota
\ref{Linfty} 
\\
Vertical/Horizontal  Lines & & Def \ref{vertical}
\\
$T_P (C)$ & & Nota \ref{nottgtes}
\\
 $\Gamma_1, \Gamma_2, \Gamma_3$  & &    Defs \ref{defC1}, \ref{defC2}, \ref{defC3}
\\
 $\C_1 (q), \C_2 (q), \C_3 (q)$  & &    Defs \ref{defC1}, \ref{defC2}, \ref{defC3}
\\
$\B$  & &      Def \ref{BBu}
 \\
$\mathcal{E}$  & &      Def \ref{BBu}
\\
$\wC_i (q)$  & &      Def \ref{defwCi}
 \\
$\Bc_i (q)$  & &      Defs \ref{I1}, \ref{I2} and \ref{I3}
\\
$\I_i (q)$ & & Defs \ref{I1}, \ref{I2} and \ref{I3}
\\
$\Pc_i (q)$  & &      Defs \ref{I1}, \ref{I2} and \ref{I3}
 \\
$E_P$ & & Nota \ref{EP}
\\
$(C3)$ configuration  & &      Def \ref{C3}
\\
$\kappa_i (q,P,L,C)$ & & Prop \ref{MultiInc}
\\
$\sim$ & & Nota \ref{para}
\\
$[L]$ & & Nota \ref{para}
\end{tabular}
\end{center}

\section*{Acknowledgements}
The author expresses a deep gratitude to Daniel Augot and Gilles Z\'emor for many inspiring discussions.
Computations and simulations have been made using \textsc{Magma}.

\bibliographystyle{abbrv}

\end{document}